\documentclass[sn-mathphys-num]{sn-jnl}

\usepackage{graphicx}%
\usepackage{multirow}%
\usepackage{amsmath,amssymb,amsfonts}%
\usepackage{amsthm}%
\usepackage{mathrsfs}%
\usepackage[title]{appendix}%
\usepackage{xcolor}%
\usepackage{textcomp}%
\usepackage{manyfoot}%
\usepackage{booktabs}%
\usepackage{algorithm}%
\usepackage{algorithmicx}%
\usepackage{algpseudocode}%
\usepackage{listings}%
\usepackage{mathtools}
\usepackage{enumerate}
\usepackage{tabularx}

\newtheorem{theorem}{Theorem}[section]
\newtheorem{definition}[theorem]{Definition}
\newtheorem{example}[theorem]{Example}

\newtheorem{remark}[theorem]{Remark}
\newtheorem{lemma}[theorem]{Lemma}
\newtheorem{corollary}[theorem]{Corollary}
\newtheorem{proposition}[theorem]{Proposition}

\DeclarePairedDelimiter\ceil{\lceil}{\rceil}
\DeclarePairedDelimiter\floor{\lfloor}{\rfloor}
\def\la{\lambda}
\def\a{\alpha}

\def\P{\mathbf P}
\def\N{\mathbb N}

\def\Z{\mathbb Z}

\def\fq{\mathbb F_q}

\def\supp{{\rm Supp}}

\def\deg{{\rm deg}}

\def\supp{{\rm supp}}
\def\res{{\rm res}}

\begin{document}

\title[Pure Gaps at Many Places and Multi-point AG Codes from Arbitrary Kummer Extensions]{Pure Gaps at Many Places and Multi-point AG Codes from Arbitrary Kummer Extensions}


\author[1]{\fnm{Huachao} \sur{Zhang}}\email{zhanghch56@mail2.sysu.edu.cn}

\author*[1,2]{\fnm{Chang-An} \sur{Zhao}}\email{zhaochan3@mail.sysu.edu.cn}

\affil[1]{\orgdiv{School of Mathematics}, \orgname{Sun Yat-sen University}, \orgaddress{\city{Guangzhou}, \postcode{510275}, \state{Guangdong Province}, \country{China}}}

\affil[2]{\orgname{Guangdong Key Laboratory of Information Security Technology}, \orgaddress{\city{Guangzhou}, \postcode{510006}, \state{Guangdong Province}, \country{China}}}


\abstract{For a Kummer extension defined by the affine equation $y^{m}=\prod_{i=1}^{r} (x-\a_i)^{\lambda_i}$ over
	an algebraic extension $K$ of a finite field $\fq$, where $\la_i\in \Z\backslash\{0\}$ for $1\leq i\leq r$, $\gcd(m,q) = 1$, and $\a_1,\cdots,\a_r\in K$ are pairwise distinct elements,
	we propose a simple and efficient method to find all pure gaps at many totally ramified places.
	We introduce a bottom set of pure gaps and indicate that the set of pure gaps is completely determined by the bottom set.
	Furthermore, we demonstrate that a pure gap can be deduced from a known pure gap by easily verifying only one inequality.
	Then, in the case where $\lambda_1 = \lambda_2 = \cdots = \lambda_r$, we fully determine an explicit description of the set of pure gaps at many totally ramified places,
	This includes the scenario in which the set of these places contains the infinite place.
	Finally, we apply these results to construct multi-point algebraic geometry codes with good parameters.
	As one of the examples, a presented code with parameters $[74, 60, \geq 10]$ over $\mathbb{F}_{25}$ yields a new record.}




\maketitle

\section{Introduction}\label{sec1}
Since Goppa introduced the algebraic geometry (AG) codes \cite{goppaCodesAssociatedDivisors1977},
the research on AG codes has drawn significant attention in the field of coding theory.
Goppa's construction of AG codes relies on a function field over a finite field.
By utilizing a function field, we hope to construct AG codes with desirable parameters, namely the length, the dimension, and the minimum distance.

One approach to constructing multi-point AG codes with good parameters is to utilize pure gaps at several places.
Homma and Kim \cite{hommaGoppaCodesWeierstrass2001} introduced the concept of pure gaps at two places, which can be used to improve the minimum distance of an AG code.
Then, this result was extended to several places by Carvalho and Torres in \cite{carvalhoGoppaCodes2005}.
The study of pure gaps and the construction of AG codes using pure gaps can be found in 
\cite{huMultipointCodes2020,carvalhoGoppaCodes2005,huMultipointCodesKummer2018,bartoliPureGapsCurves2018,tizziottiWeierstrassSemigroupPure2018,castellanosSetPure2024,castellanosOneTwoPointCodes2016,hommaGoppaCodesWeierstrass2001,castellanosWeierstrassSemigroupsPure2024,bartoliAlgebraicGeometricCodes2018,matthewsWeierstrassPairsMinimum2001,castellanosWeierstrassSemigroup2020,tenorioWeierstrassGapsSeveral2019,yangPureWeierstrassGaps2018,filhoWeierstrassPureGaps2022}.

Let $q$ be a prime power and $\fq$ a finite field with cardinality $q$. Let $K$ be an algebraic extension of $\fq$.
Consider the function field $F = K(x,y)/K(x)$ which is a Kummer extension defined by the affine equation
\begin{equation}\label{equation1}
y^m = \prod_{i=1}^r (x-\a_i)^{\la_i},
\end{equation}
where  $\la_i\in \Z\backslash\{0\}$ for $1\leq i \leq r$, $\gcd(m,q) = 1$, and $\a_1,\cdots,\a_r\in K$ are pairwise distinct elements.
Pure gaps at several totally ramified places over Kummer extensions have been extensively studied in the literature.
However, the majority of the results were based on  Kummer extensions defined by (\ref{equation1}), where all the multiplicities $\lambda_1,\cdots,\lambda_r$ are equal.
For example, Castellanos, Masuda and Quoos \cite{castellanosOneTwoPointCodes2016} characterized the set of pure gaps at two totally ramified places.
In \cite{huMultipointCodesKummer2018}, Hu and Yang provided an explicit characterization of pure gaps at several totally ramified places,
and they gave an arithmetic characterization of the set of pure gaps at several places on a quotient of the Hermition function field in \cite{yangPureWeierstrassGaps2018}.
Bartoli, Quoos and Zini \cite{bartoliAlgebraicGeometricCodes2018} proposed a criterion for determining pure gaps and presented some families of pure gaps at many totally ramified places.

Recently, pure gaps over Kummer extensions defined by (\ref{equation1}), where not all the multiplicities $\la_1,\cdots,\la_r$ are equal, have been explored.
In \cite{bartoliPureGapsCurves2018}, Bartoli et al. extended the concept of pure gaps to \textbf{c}-gaps and obtained numerous families of pure gaps at two totally ramified places.
In \cite{castellanosWeierstrassSemigroupsPure2024}, based on the knowledge of the minimal generating set of the Weierstrass semigroup,
Castellanos, Mendoza and Quoos characterized the set of pure gaps at two totally ramified places.
However, these studies only foucsed on investigating  pure gaps at two ramified places over Kummer extensions.

Explicitly describing and identifying the set of pure gaps at more than two places is a complex task.
In \cite{castellanosCompleteSetPure2024} and \cite{castellanosSetPure2024}, using maximal elements in generalized Weierstrass semigroups, Castellanos, Mendoza and Tizziotti provided a way to completely determine the set of pure gaps
at two places and several rational places in an arbitrary function field, respectively.
For Kummer extensions defined by (\ref{equation1}), where all the multiplicities $\lambda_1,\cdots,\lambda_r$ are equal, they provided an explicit description of the set of pure gaps at several totally ramified places distinct to the infinite place.
However, an explicit description of the set of pure gaps at several totally ramified places including the infinite place remains unexplored in the literature.

In this work, we focus on investigating pure gaps at many totally ramified places over Kummer extensions defined by (\ref{equation1}), where not all $\la_1,\cdots,\la_r$ are equal.
Inspired by \cite[Theorem 3.3]{bartoliPureGapsCurves2018}, we propose a new characterization for identifying pure gaps at many totally ramified places.
Then we present a straightforward and efficient method to find all pure gaps.
Our approach differs from the methods mentioned in the preceding literature.
Firstly, we introduce a bottom set of pure gaps, which has the ability to determine the set of pure gaps (see Definition \ref{def_bottom_set} and Theorem \ref{bottom_set}).
Subsequently, we propose that the bottom set of pure gaps at $s+1$ totally ramified places can be efficiently deduced from that of pure gaps at $s$ totally ramified places (see Algorithm \ref{alg2}).
Consequently, by leveraging the principle of induction and the bottom set of pure gaps, we are able to identify all pure gaps at many totally ramified places (see Algorithm \ref{alg3}).
Furthermore, we indicate that a pure gap can be deduced from a known pure gap by easily verifying only one inequality (see Theorem \ref{find_another_pure_gap}).

In particular, for Kummer extensions defined by (\ref{equation1}), where all the multiplicities $\lambda_1,\cdots,\lambda_r$ are equal,
we present a novel description of the set of gaps at some totally ramified place (see Proposition \ref{novel_gap}).
With this new characterization, we fully determine an explicit description of the set of pure gaps at many totally ramified places,
including the scenario where the set of these places contains the infinite place (see Propositons \ref{precise_puregap1}, \ref{precise_puregap2} and \ref{precise_puregap3}).
These extend the results in \cite{castellanosSetPure2024}.

Consecutive pure gaps play a crucial role in improving the minimum distance when constructing multi-point algebraic geometry codes.
Based on our methods, we easily obtain several families of consecutive pure gaps at many totally ramified places (see Propositons \ref{precisefamily1}, \ref{precisefamily2} and \ref{precisefamily3}).
By applying these consecutive pure gaps, we construct multi-point AG codes with good parameters, which outperform the parameters provided in MinT's Tables \cite{mintOnlineDatabaseOptimal}.
A new record-giving $[74,60,\geq10]$-code over $\mathbb{F}_{25}$ is presented as one of the examples (see Example \ref{new_record_example}).

The remainder of the paper is organized as follows.
In Section \ref{sec2}, we briefly review some notations and preliminary results on function fields, AG codes and pure gaps.
In Section \ref{sec3}, we present a new characterization of pure gaps at many totally ramified places and propose a simple and efficient method for finding all pure gaps at many totally ramified places.
In Section \ref{sec4}, we determine an explicit description of set of pure gaps at many totally ramified places from special Kummer extensions and obtain several families of consecutive pure gaps.
In Section \ref{sec5}, we apply our results to construct multi-point AG codes with good parameters and showcase several examples.

\section{Preliminaries}\label{sec2}
Let $q$ be a prime power and $\fq$ a finite field of cardinality $q$.
Denote by $F$ a function field and by $\P(F)$ the set of places of $F$.
Let $K$ be an algebraic extension of $\fq$. For a function field $F/K$ with full constant field $K$,
a place $Q\in \P(F)$ is said to be rational if it has degree one. Denote by $\textbf{P}_1(F)$ the set of rational places of $F$.
Denote by $\mathcal{D}(F)$ the free abelian group generated by the places of $F$. An element in $\mathcal{D}(F)$ is called a divisor.
A divisor can be written in the form $D = \sum_{P\in\P(F)}n_P P$, where $n_P\in\Z$ and $n_P = 0$ for almost all $P$.
The degree of $D$ is defined as $\deg(D) = \sum_{P\in\P(F)} n_P$, and the support set of $D$ is defined as $\supp(D)=\{P\in\P(F)~|~n_p\neq 0\}$.
For a function $z\in F$, denote by $(z)$ the principal divisor of $z$ and by $(z)_\infty$ the pole divisor of $z$.
For a divisor $D\in\mathcal{D}(F)$, the Riemann-Roch vector space with respect to $D$ is defined as
$$\mathcal{L}(D) = \{z\in F~|~(z)+D\geq 0\}\cup\{0\}.$$
Let $\ell(D)$ be the dimension of $\mathcal{L}(D)$ over $K$.

Let $D = Q_1+\cdots Q_n$ be a divisor of a function field $F$ with genus $g$ such that $Q_1,\cdots,Q_n$ are distinct rational places.
Consider another divisor $G$ of $F$ such that $\supp(D)\cap  \supp(G) = \varnothing$.
The differential algebraic geometry code $C_{\Omega}(D,G)$, which depends on the space of differentials $\Omega(G-D)$, is defined as
$$C_{\Omega}(D,G) = \{(\res_{Q_1}(\omega),\cdots,\res_{Q_n}(\omega)) ~|~ \omega\in \Omega(G-D)\}.$$
Then $C_{\Omega}(D,G)$ is an $[n,k_\Omega,d_\Omega]$ code with parameters $k_\Omega = \ell(W+D-G)-\ell(W-G)$ and $d_\Omega \geq \deg(G) - (2g-2)$,
where $W$ is a canonical divisor.
Furthermore, if $2g-2<\deg(G)<n$, then 
$$k_\Omega  = n+g-1-\deg(G).$$
For more information regarding algebraic geometry codes, the reader is referred to \cite{stichtenothAlgebraicFunctionFields2009}.

Let $\N$ be the set of non-negative integers. 
For distinct rational places $Q_1,\cdots,Q_s\in\P_1(F)$, the Weierstrass semigroup at $Q_1,\cdots,Q_s$ is defined as
$$ H(Q_1,\cdots,Q_s) = \{(a_1,\cdots,a_s)\in\N^s ~|~ (z)_\infty  = a_1Q_1+\cdots+a_sQ_s \text{ for some } z\in F\}.$$
The complement $G(Q_1,\cdots,Q_s) = \N^s\backslash H(Q_1,\cdots,Q_s)$ is always a finite set and is called the set of gaps at $Q_1,\cdots, Q_s$. An element $(a_1,\cdots,a_s)\in G(Q_1,\cdots,Q_s)$ is called a gap at $Q_1,\cdots,Q_s$.
The dimension of Riemann - Roch vector spaces can be utilized to characterize gaps, as presented in \cite{carvalhoGoppaCodes2005}. 
An $s$-tuple $(a_1,\cdots,a_s)\in \N^s$ is a gap at $Q_1,\cdots,Q_s$ if and only if
$$\ell\left(\sum_{i=1}^s a_i Q_i\right) = \ell\left(\sum_{i=1}^s a_i Q_i - Q_j\right) \text{ for some } 1\leq j\leq s.$$

Subsequently, we introduce the concept of pure gaps, which will be applied in the construction of  differential algebraic geometry codes.
An $s$-tuple $(a_1,\cdots,a_s)\in\N^s$ is a pure gap at $Q_1,\cdots, Q_s$ if and only if
$$\ell\left(\sum_{i=1}^s a_i Q_i\right) = \ell\left(\sum_{i=1}^s a_i Q_i - Q_j\right) \text{ for all } 1\leq j\leq s.$$
The set of pure gaps at $Q_1,\cdots,Q_s$ is denoted as $G_0(Q_1,\cdots,Q_s)$.
When $s=1$, it is obvious that $G(Q_1) = G_0(Q_1)$. Moreover, the gap set $G(Q_1)$ contains $g$ elements, and its elements $n_1,\cdots,n_g$ satisfy $1=n_1<\cdots<n_g\leq 2g-1$.

\begin{proposition}\label{subpuregap} 
	Suppose that $s\geq 2$, $1\leq t\leq s$, and $Q_1,\cdots,Q_s$ are the $s$ distinct rational places in the function field $F/K$.
	Let $(a_1,\cdots,a_s)\in G_0(Q_1,\cdots,Q_s)$ and $\{i_1,\cdots,i_t\}\subseteq \{1,\cdots,s\}$.
	Then $(a_{i_1},\cdots,a_{i_t})\in G_0(Q_{i_1},\cdots,Q_{i_t})$.
\end{proposition}
\begin{proof}
	Assume that $\ell\left(\sum_{j=1}^{t}a_{i_j}Q_{i_j}\right) \neq \ell\left(\sum_{j=1}^{t}a_{i_j}Q_{i_j}-Q_{i_k}\right)$ for some $1\leq k\leq t.$
	Then there exists $z\in \mathcal{L}\left(\sum_{j=1}^{t}a_{i_j}Q_{i_j}\right)$ such that $z\in \mathcal{L}\left(\sum_{i=1}^{s}a_{i}Q_{i}\right)\setminus\mathcal{L}\left(\sum_{i=1}^{s}a_{i}Q_{i}-Q_{i_k}\right)$.
	Then $\ell\left(\sum_{i=1}^{s}a_{i}Q_{i}\right)\neq\ell\left(\sum_{i=1}^{s}a_{i}Q_{i}-Q_{i_k}\right)$, which contradicts $(a_1,\cdots,a_s)\in G_0(Q_1,\cdots,Q_s)$.
\end{proof}

The following theorem shows that pure gaps can be employed to improve the lower bound for the minimum distance of a differential algebraic geometry code.
\begin{theorem}[{\cite[Theorem 3.4]{carvalhoGoppaCodes2005}}] \label{construction}
	Let $P_1, \cdots , P_n, Q_1, \cdots, Q_s$ be pairwise distinct rational places of $F$,
	and let $(a_1, \dots , a_s),(b_1, \dots , b_s)\in\N^s$ be two pure gaps at $Q_1, \cdots, Q_s$.
	Consider the divisors $D= P_1+ \cdots + P_n$ and $G= \sum_{i=1}^s (a_i+b_i-1)Q_i$.
	Suppose that $a_i \leq b_i$ for all $i=1, \cdots , s$, and that each $s$-tuple $(c_1, \cdots , c_s)\in\N^s$ such that $a_i \leq c_i \leq b_i$ for $1\leq i\leq s$, is also a pure gap at $Q_1, \cdots , Q_s$.
	Then the minimum distance $d_\Omega$ of $C_{\Omega}(D,G)$ satisfies 
	$$ d_\Omega \geq \deg(G) -(2g-2)+s+\sum_{i=1}^s (b_i-a_i).$$
	\end{theorem}

For $a, b\in \mathbb{Z}$, we denote by $\gcd(a,b)$ the greatest commom divisor of $a$ and $b$.
For $c\in \mathbb{R}$, we denote by $\floor{c}$ the largest integer not greater than $c$ and by $\ceil{c}$ the smallest integer not less than $c$.
The following two facts hold for $\floor{c}$ and $\ceil{c}$:

(i) $\floor{-c} = -\ceil{c}$;

(ii) $\ceil{c} = \floor{c}+1$ if $c\notin\Z$.

\begin{lemma}[{\cite[Lemma 7]{morenoExplicitNonspecial2024}}]\label{floor_minus}
	Let $m$ and $r$ be two positive integers such that $\gcd(m,r) = 1$. Then for $1\leq j\leq m-1$,
	$$\floor*{\frac{r(j+1)}{m}} - \floor*{\frac{rj}{m}} \geq \floor*{\frac{r}{m}}.$$
\end{lemma}

Consider the Kummer extension $F = K(x,y)/K(x)$ defined by (\ref{equation1})

$$y^m = \prod_{i=1}^r (x-\a_i)^{\la_i},$$
where  $\la_i\in \Z\setminus\{0\}$ for all $1\leq i \leq r$, $\gcd(m,q) = 1$, and $\a_1,\cdots,\a_r\in K$ are pairwise distinct elements.
Let $\la_{0} = -\sum_{i=1}^r \la_i$.
Then, according to \cite[Proposition 3.7.3]{stichtenothAlgebraicFunctionFields2009}, the genus of the function field $F/K$ is given by  
$$ g = \frac{m(r-1)+2-\sum_{i=0}^{r}\gcd(m,\la_i)}{2}.$$
It should be noted that $\sum_{i=0}^r \lambda_i  = 0$. We present a highly useful lemma. 
\begin{lemma}\label{sum_mod}
	Suppose that $(a_0,\cdots,a_r)\in \Z^{r+1}$, $b, c\in\Z$ and $b \equiv c\pmod m$.
	Then $\sum_{i=0}^r\left\lfloor\frac{a_i+b\la_i}{m}\right\rfloor = \sum_{i=0}^r\left\lfloor\frac{a_i+c\la_i}{m}\right\rfloor $.
\end{lemma}
\begin{proof}
	Suppose that $b = c + km$ for some $k\in \Z$. Then
\begin{align*}
\sum_{i=0}^r\left\lfloor\frac{a_i+b\la_i}{m}\right\rfloor
	&= \sum_{i=0}^r\left\lfloor\frac{a_i+(c+km)\la_i}{m}\right\rfloor\\
	&= \sum_{i=0}^r\left\lfloor\frac{a_i+c\la_i}{m}\right\rfloor +k\sum_{i=0}^{r}\la_i 
	 = \sum_{i=0}^r\left\lfloor\frac{a_i+c\la_i}{m}\right\rfloor.
\end{align*}
\end{proof}

For each $1\leq i\leq r$, let $P_i$ and $P_\infty$ denote the places in $\P(K(x))$ corresponding to the zero of $x-\a_i$ and the pole of $x$, respectively.
If $\gcd(m,\la_i) = 1$, we denote by $Q_i$ the only place in $\P(F)$ lying over $P_i$ and say that $Q_i$ is totally ramified.
If $\gcd(m,\la_{0}) = 1$, we denote by $Q_\infty$ the only place in $\P(F)$ lying over $P_\infty$ and say that  $Q_\infty$ is totally ramified.
For convenience, if $\gcd(m,\la_{0}) = 1$, we also denote by $P_0$ the pole of $x$ in $\P(K(x))$ and denote by $Q_0$ the only place in $\P(F)$ lying over $P_0$.

If $Q_i$ is totally ramified for some $1\leq i\leq r$, then $Q_i$ is rational and 
$$(x-\alpha_i)^{-k} = \frac{km}{\gcd(m,\la_0)}\sum_{Q\in \P(F), Q|P_\infty} Q - kmQ_i \text{ for all } k\in \N.$$
Thus $km\not\in G(Q_i)$ for all $k\in \N$.
If $Q_\infty$ is totally ramified, then for each $1\leq i\leq r$,  $Q_\infty$ is rational and 
$$(x-\alpha_i)^k = \frac{km}{\gcd(m,\la_i)}\sum_{Q\in \P(F), Q|P_i} Q - kmQ_\infty \text{ for all } k\in \N.$$
Thus $km\not\in G(Q_i)$ for all $k\in \N$.

In particular, if $\lambda_1 = \lambda_2 = \cdots = \lambda_r$, one can explicitly describe the set $G(Q_i)$ for $0\leq i\leq r$. 
\begin{lemma}[{\cite[Corollary 2]{huMultipointCodesKummer2018}}]\label{special_gap}
	Suppose that $\lambda_1 = \lambda_2 = \cdots = \lambda_r$ and $\gcd(r\lambda_1, m)  =1$. Then for each $1\leq i\leq r$,
	$$G(Q_i) = \left\{ mk+j ~|~ 1\leq j\leq m-1-\floor*{\frac{m}{r}},~ 0\leq k \leq r-2-\floor*{\frac{rj}{m}} \right\}.$$
	For the infinite place $Q_\infty$,
	$$G(Q_\infty) =  \left\{ mk-rj ~|~ 1\leq j\leq m-1-\floor*{\frac{m}{r}},~ \ceil*{\frac{rj}{m}}\leq k \leq r-1 \right\}.$$
\end{lemma}

Let $\{j_0,j_1\cdots,j_r\}$ be a permutation of $\{0,1,\cdots,r\}$.
Let $Q_{j_0}, Q_{j_1},\cdots,Q_{j_s}$ be $s+1$ totally ramified places.
For $0\leq i\leq s$, note that $km\not\in G(Q_{j_i})$ for all $k\in \N$.
Given $(a_0,a_1,\cdots,a_s)\in G_0(Q_{j_0}, Q_{j_1},\cdots,Q_{j_s})$, according to Proposition \ref{subpuregap}, we conclude that $a_i\not\equiv 0 \pmod m$ for all $0\leq i\leq s$.
Based on this observation, we introduce the concept of bottom sets of pure gaps at many totally ramified places.

\begin{definition}\label{def_bottom_set}
	Suppose that $0\leq s \leq r$ and $\gcd(m,\la_{j_i})=1$ for all $0\leq i \leq s$. The bottom set of pure gaps at  $Q_{j_0},Q_{j_1},\cdots, Q_{j_s}$ is defined as
	$$\overline{G}_0(Q_{j_0},Q_{j_1},\cdots, Q_{j_s}) = G_0(Q_{j_0},Q_{j_1},\cdots, Q_{j_s})\cap [1,m-1]^{s+1},$$
	where $[1,m-1]^{s+1} = \{(a_0,a_1,\cdots,a_s)\in \N^{s+1}~|~1\leq a_i\leq m-1 \text{ for } 0\leq i\leq s\}$.
	In particular, if $s = 0$, the bottom set of gaps at $Q_{j_0}$ is defined as 
	$$\overline{G}(Q_{j_0}) = G(Q_{j_0})\cap\{1,2,\cdots,m-1\}.$$
\end{definition}

The following theorem will serve as a critical foundation for investigating how to identify pure gaps at several totally ramified places.
\begin{theorem}[{\cite[Theorem 3.3]{bartoliPureGapsCurves2018}}]\label{puregapsmanypoints}
	Suppose that $0\leq s \leq r$ and $\gcd(m,\la_{j_i})=1$ for all $0\leq i \leq s$.
	Then $(a_0,a_1,\cdots,a_s)\in \N^{s+1}$ is a pure gap at $Q_{j_0},Q_{j_1},\cdots, Q_{j_s}$ if and only if,
	for each $t\in\{0,1,\cdots, m-1\}$ exactly one of the following two conditions is satisfied:
	
	(i) $\sum_{i=0}^{s}\left\lfloor\frac{a_i+t\la_{j_i}}{m}\right\rfloor + \sum_{i=s+1}^{r}\left\lfloor\frac{t\la_{j_i}}{m}\right\rfloor< 0$;

	(ii) $\sum_{i=0}^{s}\left\lfloor\frac{a_i+t\la_{j_i}}{m}\right\rfloor + \sum_{i=s+1}^{r}\left\lfloor\frac{t\la_{j_i}}{m}\right\rfloor \geq 0$ and $\left\lfloor\frac{a_i+t\la_{j_i}}{m}\right\rfloor = \left\lfloor\frac{a_i+t\la_{j_i}-1}{m}\right\rfloor$ for all $0\leq i\leq s$.
\end{theorem}

\section{Pure gaps from arbitrary Kummer extensions}\label{sec3}

Throughout this section, $F = K(x,y)/K(x)$ is a Kummer extension defined by (\ref{equation1}).
For each $1\leq i\leq r$, let $P_i$ and $P_\infty$ denote the places in $\P(K(x))$ corresponding to the zero of $x-\a_i$ and the pole of $x$, respectively.
If $\gcd(m,\la_i) = 1$ for $1\leq i\leq r$, we denote by $Q_i$ the only place in $\P(F)$ lying over $P_i$.
For convenience, if $\gcd(m,\la_{0}) = 1$, we denote by $P_0$ the pole of $x$ in $\P(K(x))$ and denote by $Q_0$ the only place in $\P(F)$ lying over $P_0$.
Let $\{j_0,j_1\cdots,j_r\}$ be a permutation of $\{0,1,\cdots,r\}$.

\subsection{A new characterization of pure gaps}
In this subsection, inspired by Theorem \ref{puregapsmanypoints}, we point out that the set of pure gaps is completely determined by the bottom set through the following theorem.
This theorem provides a new characterization for identifying pure gaps at many totally ramified places.

\begin{theorem}\label{bottom_set}
	Suppose that $0\leq s \leq r$ and $\gcd(m,\lambda_{j_i})=1$ for all $0\leq i \leq s$.
	Let $\sigma_{j_i}\in \Z$ be the inverse of $\lambda_{j_i}$ modulo $m$ for $0\leq i\leq s$.
	For $\mathbf{a} = (a_0,a_1,\cdots,a_s)\in \overline{G}_0(Q_{j_0},Q_{j_1},\cdots, Q_{j_s})$, define
	$$k_{\mathbf{a}} = -1-\max_{0\leq v\leq s}\left\{\sum_{i=0}^{s}\left\lfloor\frac{a_i-a_v\sigma_{j_v}\la_{j_i}}{m}\right\rfloor - \sum_{i=s+1}^{r}\left\lfloor\frac{-a_v\sigma_{j_v}\la_{j_i}}{m}\right\rfloor\right\}.$$
	Then
	\begin{align*}
	&G_0(Q_{j_0},Q_{j_1},\cdots, Q_{j_s}) = \Big\{(mk_0+a_0,mk_1+a_1,\cdots,mk_s+a_s)~\Big|\\
	&\mathbf{a} = (a_0,a_1,\cdots,a_s)\in\overline{G}_0(Q_{j_0},Q_{j_1},\cdots, Q_{j_s}),~k_i\geq 0 \text{ for } 0\leq i\leq s, \text{ and } \sum_{i=0}^s k_i \leq k_{\mathbf{a}}\Big\}.
	\end{align*}
\end{theorem}

In order to prove Theorem \ref{bottom_set}, we propose a simple criterion for identifying pure gaps at many totally ramified places.
\begin{lemma}\label{simple_criterion}
	Suppose that $0\leq s \leq r$ and $\gcd(m,\la_{j_i})=1$ for all $0\leq i \leq s$.
	Let $\sigma_{j_i}\in \Z$ be the inverse of $\lambda_{j_i}$ modulo $m$ for $0\leq i\leq s$.
	Then $(a_0,a_1,\cdots,a_s)\in \N^{s+1}$ is a pure gap at $Q_{j_0},Q_{j_1},\cdots, Q_{j_s}$ if and only if,
	for each $0\leq k\leq s$ the following condition holds:
	$$\sum_{i=0}^s\floor*{\frac{a_i-a_k\sigma_{j_k}\lambda_{j_i}}{m}} + \sum_{i=s+1}^{r}\floor*{\frac{-a_k\sigma_{j_k}\lambda_{j_i}}{m}}\leq -1.$$
\end{lemma}
\begin{proof}
	For each $0\leq k\leq s$, note that $\floor*{\frac{a_k+t\lambda_{j_k}}{m}}\neq \floor*{\frac{a_k+t\lambda_{j_k}-1}{m}}$ if and only if $t\equiv -a_k\sigma_{j_k} \pmod m$.
	Assume that  $(a_0,a_1,\cdots,a_s)$ is a pure gap at $Q_{j_0},Q_{j_1},\cdots, Q_{j_s}$. 
	Then, for each $t\in\{0,1,\cdots, m-1\}$, exactly one of two conditions in Theorem \ref{puregapsmanypoints} is satisfied.
	Thus for $t\in\{t~|~t\equiv -a_k\sigma_{j_k}\pmod m,~0\leq k\leq s\}\cap\{0,1,\cdots,m-1\}$, it follows from Lemma \ref{sum_mod} that
	\begin{align*}
		\sum_{i=0}^s\floor*{\frac{a_i-a_k\sigma_{j_k}\lambda_{j_i}}{m}} + \sum_{i=s+1}^{r}\floor*{\frac{-a_k\sigma_{j_k}\lambda_{j_i}}{m}} = \sum_{i=0}^s\floor*{\frac{a_i+t\lambda_{j_i}}{m}} + \sum_{i=s+1}^{r}\floor*{\frac{t\lambda_{j_i}}{m}}\leq -1.
	\end{align*}
	Now assume that 
	$$\sum_{i=0}^s\floor*{\frac{a_i-a_k\sigma_{j_k}\lambda_{j_i}}{m}} + \sum_{i=s+1}^{r}\floor*{\frac{-a_k\sigma_{j_k}\lambda_{j_i}}{m}}\leq -1 \text{ for all } 0\leq k\leq s.$$
	Then, for $t\in\{t~|~t\equiv -a_k\sigma_{j_k}\pmod m,~0\leq k\leq s\}\cap\{0,1,\cdots,m-1\}$, by Lemma \ref{sum_mod}, we obtain
	\begin{align*}
		\sum_{i=0}^s\floor*{\frac{a_i+t\lambda_{j_i}}{m}} + \sum_{i=s+1}^{r}\floor*{\frac{t\lambda_{j_i}}{m}} = \sum_{i=0}^s\floor*{\frac{a_i-a_k\sigma_{j_k}\lambda_{j_i}}{m}} + \sum_{i=s+1}^{r}\floor*{\frac{-a_k\sigma_{j_k}\lambda_{j_i}}{m}} \leq -1.
	\end{align*}
	For $t\in\{t~|~t\not\equiv -a_k\sigma_{j_k}\pmod m,~0\leq k\leq s\}\cap\{0,1,\cdots,m-1\}$, either $\sum_{i=0}^{s}\left\lfloor\frac{a_i+t\la_{j_i}}{m}\right\rfloor + \sum_{i=s+1}^{r}\left\lfloor\frac{t\la_{j_i}}{m}\right\rfloor< 0$ or 
	$\left\lfloor\frac{a_i+t\la_{j_i}}{m}\right\rfloor = \left\lfloor\frac{a_i+t\la_{j_i}-1}{m}\right\rfloor$ for all $0\leq i\leq s$.
	Thus, by Theorem \ref{puregapsmanypoints}, we have $(a_0,a_1,\cdots,a_s)$ is a pure gap at $Q_{j_0},Q_{j_1},\cdots, Q_{j_s}$.
\end{proof}

Now we are in a position to give the proof of Theorem \ref{bottom_set}.
\begin{proof}[Proof of Theorem \ref{bottom_set}]
	We set $G = \{(mk_0+a_0,mk_1+a_1,\cdots,mk_s+a_s)~|~\mathbf{a} = (a_0,a_1,\cdots,a_s)\in\overline{G}_0(Q_{j_0},Q_{j_1},\cdots, Q_{j_s}),$ $k_i\geq 0 \text{ for } 0\leq i\leq s, \text{ and } \sum_{i=0}^s k_i \leq k_{\mathbf{a}}\}$
	and $G_0 = G_0(Q_{j_0},Q_{j_1}$, $\cdots, Q_{j_s})$.
	Given $\mathbf{b} = (b_0,b_1,\cdots,b_s)\in G_0$, for each $0\leq i\leq s$, we write $b_i = mk_i+c_i$ with $k_i\geq 0$ and $1\leq c_i\leq m-1$.
	Let $\mathbf{c} = (c_0,c_1,\cdots,c_s)$.
	By Lemmas \ref{sum_mod} and \ref{simple_criterion}, for each $0\leq v\leq s$, we have
	\begin{align*}
		&~\sum_{i=0}^s\floor*{\frac{c_i-c_v\sigma_{j_v}\lambda_{j_i}}{m}} + \sum_{i=s+1}^{r}\floor*{\frac{-c_v\sigma_{j_v}\lambda_{j_i}}{m}}+\sum_{i=0}^s k_i\\
		=&~\sum_{i=0}^s\floor*{\frac{b_i-b_v\sigma_{j_v}\lambda_{j_i}}{m}} + \sum_{i=s+1}^{r}\floor*{\frac{-b_v\sigma_{j_v}\lambda_{j_i}}{m}}\leq -1.
	\end{align*}
	Thus, for each $0\leq v\leq s$, we have
	$$\sum_{i=0}^s\floor*{\frac{c_i-c_v\sigma_{j_v}\lambda_{j_i}}{m}} - \sum_{i=s+1}^{r}\floor*{\frac{-c_v\sigma_{j_v}\lambda_{j_i}}{m}}\leq -\sum_{i=0}^s k_i -1.$$
	It follows from Lemma \ref{simple_criterion} that $\mathbf{c}\in \overline{G}_0(Q_{j_0},\cdots,Q_{j_s})$ and $\sum_{i = 0}^s k_i\leq k_{\mathbf{c}}$.
	Therefore, we conclude that $\mathbf{b}\in G$. This implies that $G_0\subseteq G$.

	Now let $\mathbf{b} =  (b_0,b_1,\cdots,b_s) = (mk_0+a_0,mk_1+a_1,\cdots,mk_s+a_s)\in G$.
	By Lemmas \ref{sum_mod} and \ref{simple_criterion}, for each $0\leq v\leq s$, we have
	\begin{align*}
		&~\sum_{i=0}^s\floor*{\frac{b_i-b_v\sigma_{j_v}\lambda_{j_i}}{m}} + \sum_{i=s+1}^{r}\floor*{\frac{-b_v\sigma_{j_v}\lambda_{j_i}}{m}}\\
		=&~\sum_{i=0}^s\floor*{\frac{a_i-a_v\sigma_{j_v}\lambda_{j_i}}{m}} + \sum_{i=s+1}^{r}\floor*{\frac{-a_v\sigma_{j_v}\lambda_{j_i}}{m}}+\sum_{i=0}^s k_i\\
		\leq&~ -1-k_{\mathbf{a}}+\sum_{i=0}^s k_i\leq -1.
	\end{align*}
	Consequently, it follows from Lemma \ref{simple_criterion} that $\mathbf{b}\in G_0$, and thus $G\subseteq G_0$.
	Hence, we have $G_0 = G$.
\end{proof}

\subsection{An efficient method to find all pure gaps}
In this subsection, we first provide a novel description of the set of gaps at a totally ramified place.
Then we present an efficient method for determining the bottom set of pure gaps at $s+1$ totally ramified places from the bottom set of pure gaps at $s$ totally ramified places.
Therefore, by utilizing the bottom set of gaps at a totally ramified place and the principle of induction, 
we are able to identify all pure gaps at $s + 1$ totally ramified places.
Furthermore, we demonstrate that a pure gap can be deduced from a known pure gap by easily verifying only one inequality.

We start this subsection by presenting our description of the set of gap at a totally ramified place,
which is distinct from \cite[Corollary 3.6]{abdonWeierstrassPoints2019}.

\begin{proposition}\label{gap_set}
	Suppose that $\gcd(m,\la_s)=1$ for some $0\leq s\leq r$ and let $\la \in \mathbb{Z} $ be the inverse of $\la_s$ modulo $m$.
	Then 
	$$G(Q_s) = \left\{a\in \N ~\Big |~ \sum_{i=0}^r\left\lfloor\frac{-a\la\la_i}{m}\right\rfloor + \left\lceil\frac{a}{m}\right\rceil \leq -1\right\}.$$
	
	Furthermore, suppose that $\lambda_1=\cdots=\lambda_r$ and $\gcd(m,r\lambda_1) = 1$. Then for $1\leq s\leq r$,
	$$G(Q_s) = \left\{a\in \N ~\Big |~ \left\lfloor\frac{ra}{m}\right\rfloor - (r-1)\left\lceil\frac{a}{m}\right\rceil \leq -1\right\}.$$
	Let $ r^\prime \in \mathbb{Z}$ be the inverse of $r$ modulo $m$. Then
	$$G(Q_\infty) = \left\{a\in \N ~\Big |~ r\left\lfloor\frac{ar^\prime}{m}\right\rfloor+ \left\lfloor\frac{-arr^\prime}{m}\right\rfloor+\left\lceil\frac{a}{m}\right\rceil \leq -1\right\}.$$
\end{proposition}
\begin{proof}
	According to Lemma \ref{simple_criterion}, we have $a$ is a gap at $Q_s$ if and only if 
	$$ \left\lfloor\frac{a-a\la\la_s}{m}\right\rfloor+\sum_{i=0, i\neq s}^r\left\lfloor\frac{-a\la\la_i}{m}\right\rfloor \leq -1.$$
	Let $ \la\la_s = 1 + km$ for some $k\in \Z$. Then $\left\lfloor\frac{a-a\la\la_s}{m}\right\rfloor  = -ka$ and $\left\lfloor\frac{-a\la\la_s}{m}\right\rfloor = -ka + \left\lfloor\frac{-a}{m}\right\rfloor = -ka -\left\lceil\frac{a}{m}\right\rceil$.
	Thus, we obtain that $\left\lfloor\frac{a-a\la\la_s}{m}\right\rfloor = \left\lfloor\frac{-a\la\la_s}{m}\right\rfloor + \left\lceil\frac{a}{m}\right\rceil$.
	It follows that
	$$ \left\lfloor\frac{a-a\la\la_s}{m}\right\rfloor+\sum_{i=0, i\neq s}^r\left\lfloor\frac{ a\la\la_i}{m}\right\rfloor = \sum_{i=0}^r\left\lfloor\frac{-a\la\la_i}{m}\right\rfloor +\left\lceil\frac{a}{m}\right\rceil.$$
	Therefore, we conclude that $a$ is a gap at $Q_s$ if and only if $\sum_{i=0}^r\left\lfloor\frac{-a\la\la_i}{m}\right\rfloor + \left\lceil\frac{a}{m}\right\rceil \leq -1$.
	
	In the case where $\lambda_1=\cdots=\lambda_r$ for $1\leq s\leq r$, by Lemma \ref{sum_mod}, we have 
	$$\sum_{i=0}^r\left\lfloor\frac{-a\la\la_i}{m}\right\rfloor + \left\lceil\frac{a}{m}\right\rceil = \sum_{i=1}^r\left\lfloor\frac{-a}{m}\right\rfloor + \left\lfloor\frac{ra}{m}\right\rfloor + \left\lceil\frac{a}{m}\right\rceil
	= \left\lfloor\frac{ra}{m}\right\rfloor - (r-1)\left\lceil\frac{a}{m}\right\rceil.$$
	Hence, we obtain that $a\in G(Q_s)$ if and only if $\left\lfloor\frac{ra}{m}\right\rfloor - (r-1)\left\lceil\frac{a}{m}\right\rceil \leq -1$.
	Let $\sigma\in\Z$ be the inverse of $-r\lambda_1$ modulo $m$, then $-r^\prime \equiv \sigma\lambda_1 \pmod{m}$.
	By Lemma \ref{sum_mod}, we obtain
	$$\sum_{i=0}^r \floor*{\frac{-a\sigma\lambda_i}{m}} +\ceil*{\frac{a}{m}} = \sum_{i=1}^{r}\floor*{\frac{-a\sigma\lambda_1}{m}}+\floor*{\frac{ar\sigma\lambda_1}{m}}+\ceil*{\frac{a}{m}}=r\floor*{\frac{ar^\prime}{m}}+\floor*{\frac{-arr^\prime}{m}}+\ceil*{\frac{a}{m}}.$$
	Therefore, we have $a\in G(Q_\infty)$ if and only if $r\floor*{\frac{ar^\prime}{m}}+\floor*{\frac{-arr^\prime}{m}}+\ceil*{\frac{a}{m}}\leq -1$.
\end{proof}

\begin{remark}\label{special_infinity}
	Suppose that $\la_1=\cdots = \la_r$ and $\gcd(m,r\la_1) = 1$. If $r \equiv -1 \pmod m$, then $r^\prime\equiv -1 \pmod m$. By Lemma \ref{sum_mod}, we obtain
	$$ r\left\lfloor\frac{ar^\prime}{m}\right\rfloor+ \left\lfloor\frac{-arr^\prime}{m}\right\rfloor+\left\lceil\frac{a}{m}\right\rceil = r\left\lfloor\frac{-a}{m}\right\rfloor + \left\lfloor\frac{ar}{m}\right\rfloor + \left\lceil\frac{a}{m}\right\rceil = \left\lfloor\frac{ra}{m}\right\rfloor - (r-1)\left\lceil\frac{a}{m}\right\rceil.$$
	Therefore, we have $G(Q_s) = G(Q_\infty)$ for $1\leq s \leq r$.
\end{remark}
Based on Proposition \ref{gap_set}, we are able to design Algorithm \ref{alg1} to acquire the bottom set of gaps at a totally ramified place.

\begin{algorithm}[H]
	\caption{Find the bottom set of gaps at a totally ramified place.}
	\label{alg1}
	\begin{algorithmic}[1]
		\Require A totally ramified place $Q_s$ of the function field $F$.
		\Ensure The bottom set $G$ of gaps at the place $Q_s$.
		\State $G\leftarrow \varnothing $;
		\State Compute $\la$ such that $1\leq \la\leq m-1$ and $\la\la_{s}\equiv 1 \pmod m$;
		\For {$a=1$ to $m-1$}
		\If {$\sum_{i=0}^r\left\lfloor\frac{-a\la\la_i}{m}\right\rfloor\leq -2$}
		\State $G\leftarrow G\cup \{a\}$;
		\EndIf
		\EndFor
		\State\textbf{return} $G$.
	\end{algorithmic}
\end{algorithm}

\begin{proposition}\label{find_bottom_gap}
	Suppose that $1\leq s \leq r$ and $\gcd(m,\la_{j_i})=1$ for all $0\leq i \leq s$.
	Let $(a_0,\cdots,a_{s-1})\in \overline{G}_0(Q_{j_0},\cdots, Q_{j_{s-1}})$ and $1\leq a\leq m-1$.
	Let $\la\in \Z$ be the inverse of $\la_{j_s}$ modulo $m$.
	Suppose that 
	$$\sum_{i=0}^{s-1}\left\lfloor\frac{a_i-a\la\la_{j_i}}{m}\right\rfloor + \left\lfloor\frac{a-a\la\la_{j_s}}{m}\right\rfloor+\sum_{i=s+1}^{r}\left\lfloor\frac{-a\la\la_{j_i}}{m}\right\rfloor\leq -1.$$
	Then $(a_0,\cdots,a_{s-1},a)\in  \overline{G}_0(Q_{j_0},\cdots, Q_{j_{s-1}},Q_{j_s})$ if and only if,
	for each $1\leq k\leq a-1$, at least one of the following two conditions holds:
	\begin{align}
		&\sum_{i=0}^{s-1}\left\lfloor\frac{a_i-k\la\la_{j_i}}{m}\right\rfloor + \left\lfloor\frac{a-k\la\la_{j_s}}{m}\right\rfloor+\sum_{i=s+1}^{r}\left\lfloor\frac{-k\la\la_{j_i}}{m}\right\rfloor \leq -1, \label{condition3}\\
		&a_i\not\equiv k\la\la_{j_i} \text{ for all } 0\leq i\leq s-1. \label{condition4}
	\end{align}
\end{proposition}
\begin{proof}
	Let $1\leq k\leq m$ and $\la\la_{j_s} = bm+1$ for some $b\in \Z$. Then, we have
	$\floor*{{\frac{a-k\la\la_{j_s}}{m}}} = \floor*{\frac{a-k}{m}} - kb$ and $\floor*{{\frac{-k\la\la_{j_s}}{m}}} = -kb-1$.
	This implies that $\left\lfloor\frac{a-k\la\la_{j_s}}{m}\right\rfloor = \left\lfloor\frac{-k\la\la_{j_s}}{m}\right\rfloor$ if and only if $\floor*{\frac{a-k}{m}} = -1$.
	Consequently,
	\begin{align}
		\left\lfloor\frac{a-k\la\la_{j_s}}{m}\right\rfloor = \left\lfloor\frac{-k\la\la_{j_s}}{m}\right\rfloor \text{ if and only if }  a+1\leq k\leq m. \label{equivalent0}
	\end{align}
	For $0\leq i \leq s-1$, let $\sigma_{j_i}\in\Z$ be the inverse of $\lambda_{j_i}$ modulo $m$.
	Then, the condition (\ref{condition4}) is equivalent to $a_i\sigma_{j_i}\not\equiv k\la$ for all $0\leq i\leq s-1$.
	
	Assume that $(a_0,\cdots,a_{s-1},a)\in \overline{G}_0(Q_{j_0},\cdots,Q_{j_{s-1}},Q_{j_s})$.
	For each $1\leq k\leq a-1$, if $a_l\sigma_{j_l}\equiv k\la$ for some $0\leq l\leq s-1$,
	then by Lemmas \ref{sum_mod} and \ref{simple_criterion}, we have
	\begin{align*}
		~&\sum_{i=0}^{s-1}\left\lfloor\frac{a_i-k\la\la_{j_i}}{m}\right\rfloor + \left\lfloor\frac{a-k\la\la_{j_s}}{m}\right\rfloor+\sum_{i=s+1}^{r}\left\lfloor\frac{-k\la\la_{j_i}}{m}\right\rfloor\\
		=~&\sum_{i=0}^{s-1}\left\lfloor\frac{a_i-a_l\sigma_{j_l}\la_{j_i}}{m}\right\rfloor + \left\lfloor\frac{a-a_l\sigma_{j_l}\la_{j_s}}{m}\right\rfloor+\sum_{i=s+1}^{r}\left\lfloor\frac{-a_l\sigma_{j_l}\la_{j_i}}{m}\right\rfloor \leq -1.
	\end{align*}

	Conversely, assume that, for each $1\leq k\leq a-1$, at least one of the conditions (\ref{condition3}) and (\ref{condition4}) is satisfied.
	For each $0\leq l\leq s-1$, if $a_l\sigma_{j_l}\equiv k\la$ for some $1\leq k\leq a$, then by Lemma \ref{sum_mod}, we have
	\begin{align*}
		~&\sum_{i=0}^{s-1}\left\lfloor\frac{a_i-a_l\sigma_{j_l}\la_{j_i}}{m}\right\rfloor + \left\lfloor\frac{a-a_l\sigma_{j_l}\la_{j_s}}{m}\right\rfloor+\sum_{i=s+1}^{r}\left\lfloor\frac{-a_l\sigma_{j_l}\la_{j_i}}{m}\right\rfloor\\
		=~&\sum_{i=0}^{s-1}\left\lfloor\frac{a_i-k\la\la_{j_i}}{m}\right\rfloor + \left\lfloor\frac{a-k\la\la_{j_s}}{m}\right\rfloor+\sum_{i=s+1}^{r}\left\lfloor\frac{-k\la\la_{j_i}}{m}\right\rfloor \leq -1.
	\end{align*}
	If $a_l\sigma_{j_l}\equiv k\la$ for some $a+1\leq k\leq m$, then by Lemma \ref{sum_mod} and (\ref{equivalent0}), we have 
	\begin{align*}
		~&\sum_{i=0}^{s-1}\left\lfloor\frac{a_i-a_l\sigma_{j_l}\la_{j_i}}{m}\right\rfloor + \left\lfloor\frac{a-a_l\sigma_{j_l}\la_{j_s}}{m}\right\rfloor+\sum_{i=s+1}^{r}\left\lfloor\frac{-a_l\sigma_{j_l}\la_{j_i}}{m}\right\rfloor\\
		=~&\sum_{i=0}^{s-1}\left\lfloor\frac{a_i-k\la\la_{j_i}}{m}\right\rfloor + \left\lfloor\frac{a-k\la\la_{j_s}}{m}\right\rfloor+\sum_{i=s+1}^{r}\left\lfloor\frac{-k\la\la_{j_i}}{m}\right\rfloor\\
		=~&\sum_{i=0}^{s-1}\left\lfloor\frac{a_i-k\la\la_{j_i}}{m}\right\rfloor + \left\lfloor\frac{-k\la\la_{j_s}}{m}\right\rfloor+\sum_{i=s+1}^{r}\left\lfloor\frac{-k\la\la_{j_i}}{m}\right\rfloor\\
		=~&\sum_{i=0}^{s-1}\left\lfloor\frac{a_i-a_l\sigma_{j_l}\la_{j_i}}{m}\right\rfloor + \left\lfloor\frac{-a_l\sigma_{j_l}\la_{j_s}}{m}\right\rfloor+\sum_{i=s+1}^{r}\left\lfloor\frac{-a_l\sigma_{j_l}\la_{j_i}}{m}\right\rfloor.
	\end{align*}
	Since  $(a_0,\cdots,a_{s-1})\in \overline{G}_0(Q_{j_0},\cdots, Q_{j_{s-1}})$, it follows that
	$$\sum_{i=0}^{s-1}\left\lfloor\frac{a_i-a_l\sigma_{j_l}\la_{j_i}}{m}\right\rfloor + \left\lfloor\frac{a-a_l\sigma_{j_l}\la_{j_s}}{m}\right\rfloor+\sum_{i=s+1}^{r}\left\lfloor\frac{-a_l\sigma_{j_l}\la_{j_i}}{m}\right\rfloor\leq -1.$$
	Therefore, by Lemma \ref{simple_criterion}, we have $(a_0,\cdots,a_{s-1},a)\in \overline{G}_0(Q_{j_0},\cdots,Q_{j_{s-1}},Q_{j_s})$.
\end{proof}

Based on Proposition \ref{find_bottom_gap}, we immediately derive Algorithm \ref{alg2},
which shows an efficient way to obtain the bottom set of pure gaps at $s+1$ totally ramified places from the bottom set of pure gaps at $s$ totally ramified places.
\begin{algorithm}[H]
	\caption{Find $\overline{G}_0(Q_{j_0},\cdots,Q_{j_{s-1}},Q_{j_s})$ from $\overline{G}_0(Q_{j_0},\cdots,Q_{j_{s-1}})$}
	\label{alg2}
	\begin{algorithmic}[1]
		\Require $G = \overline{G}_0(Q_{j_0},\cdots,Q_{j_{s-1}})$.
		\Ensure  $A = \overline{G}_0(Q_{j_0},\cdots,Q_{j_{s-1}},Q_{j_s})$.
		\State $A\leftarrow \varnothing $;
		\State Compute $\la$ such that $1\leq \la\leq m-1$ and $\la\la_{j_s}\equiv 1 \pmod m$;
		\For {each $\mathbf{a} = (a_0,\cdots,a_{s-1}) \in G$}
		\For{$a=1$ to $m-1$}
		\If {$\sum_{i=0}^{s-1}\left\lfloor\frac{a_i-a\la\la_{j_i}}{m}\right\rfloor + \left\lfloor\frac{a-a\la\la_{j_s}}{m}\right\rfloor+\sum_{i=s+1}^{r}\left\lfloor\frac{-a\la\la_{j_i}}{m}\right\rfloor \leq -1$}
		\State $A\leftarrow A\cup \{(a_0,a_1,\cdots,a_{s-1},a)\}$;
		\ElsIf {$a_i\equiv a\la\la_{j_i} \pmod m$ for some $0\leq i\leq s-1$}
		\State \textbf{break};
		\EndIf
		\EndFor
		\EndFor
		\State \textbf{return} $A$.
	\end{algorithmic}
\end{algorithm}

	Based on Algorithm \ref{alg2}, we can determine the bottom set of pure gaps at $s+1$ totally ramified places from the bottom set of pure gaps at $s$ totally ramified places.
	Furthermore, we can determine the bottom set of pure gaps at $s$ totally ramified places from the bottom set of pure gaps at $s-1$ totally ramified places.
	By applying the principle of induction, we can determine the bottom set of pure gaps at $s+1$ totally ramified places from the bottom set of gaps at a single totally ramified place.
	Then, by Theorem \ref{bottom_set}, we can determine the set of pure gaps at $s+1$ totally ramified places.
	We describe this process as the following Algorithm.

\begin{algorithm}[H]
	\caption{Find the set ${G}_0(Q_{j_0},Q_{j_1},\cdots,Q_{j_s})$.}
	\label{alg3}
	\begin{algorithmic}[1]
		\Require $s+1$ totally ramified places $Q_{j_0},Q_{j_1},\cdots,Q_{j_s}$.
		\Ensure The set of pure gaps $A = {G}_0(Q_{j_0},Q_{j_1},\cdots,Q_{j_s})$.
		\State Input $Q_{j_0}$ into Algorithm \ref{alg1} and output the set $G = \overline{G}(Q_{j_0})$;
		\For {$i=1$ to $s$}
		\State Input $G$ to Algorithm \ref{alg2} and output the set $A$;
		\State $G\leftarrow A$;
		\EndFor
		\State Compute $\sigma_{j_i}$ such that $1\leq \sigma_{j_i}\leq m-1$ and $\sigma_{j_i}\la_{j_i}\equiv 1 \pmod m$ for $0\leq i\leq s$;
		\State $A\leftarrow \varnothing$;
		\For {each $\mathbf{a} = (a_0,a_1,\cdots,a_s)\in G$}
		\State $k \leftarrow \max_{0\leq t\leq s}\left\{\sum_{i=0}^{s}\left\lfloor\frac{a_i-a_t\sigma_{j_t}\la_{j_i}}{m}\right\rfloor+\sum_{i=s+1}^{r}\left\lfloor\frac{-a_t\sigma_{j_t}\la_{j_i}}{m}\right\rfloor\right\}$;
		\State $A\leftarrow A\cup \{(mk_0+a_0,\cdots,mk_s+a_s)~|~k_i\geq 0 \text{ for } 0\leq i\leq s,~\sum_{i=1}^s k_i\leq -1-k\}$;
		\EndFor
		\State \textbf{return} $A$.
	\end{algorithmic}
\end{algorithm}

	Algorithm \ref{alg3} presents a straightforward and efficient method to find all pure gaps at many totally ramified places.
	Moreover, this method can identify consecutive pure gaps, which are highly valuable for constructing AG codes with good parameters in Proposition \ref{construction}.
	At the end of this section, we show that given knowledge of a pure gap, we can derive another pure gap by easily verifying only one inequality.

\begin{theorem}\label{find_another_pure_gap}
	Suppose that $1\leq s \leq r$ and $\gcd(m,\la_{j_i})=1$ for all $0\leq i \leq s$, and let $\la \in \mathbb{Z} $ be the inverse of $\la_{j_s}$ modulo $m$.
	Suppose that $(a_0,\cdots,a_{s-1})\in G_0(Q_{j_0}, \cdots, Q_{j_{s-1}})$.
	Then $(a_0,\cdots,a_{s-1},1) \in G_0(Q_{j_0},\cdots, Q_{j_{s-1}},Q_{j_s})$ if and only if
	$$\sum_{i=0}^{s-1}\left\lfloor\frac{a_i-\la\la_{j_i}}{m}\right\rfloor + \left\lfloor\frac{1-\la\la_{j_s}}{m}\right\rfloor+\sum_{i=s+1}^{r}\left\lfloor\frac{-\la\la_{j_i}}{m}\right\rfloor \leq -1.$$
	
	Furthermore, suppose that $(a_0,\cdots,a_{s-1},a) \in G_0(Q_{j_0},\cdots,Q_{j_{s-1}},Q_{j_s})$ for some $1\leq a\leq m-2$.
	Then $(a_0,\cdots,a_{s-1},a+1)\in G_0(Q_{j_0},\cdots, Q_{j_{s-1}}, Q_{j_s})$ if and only if
	$$\sum_{i=0}^{s-1}\left\lfloor\frac{a_i-(a+1)\la\la_{j_i}}{m}\right\rfloor + \left\lfloor\frac{(a+1)-(a+1)\la\la_{j_s}}{m}\right\rfloor + \sum_{i=s+1}^{r}\left\lfloor\frac{-(a+1)\la\la_{j_i}}{m}\right\rfloor \leq -1.$$
\end{theorem}
\begin{proof}
	For $0\leq i \leq s-1$, let $\sigma_{j_i}\in\Z$ be the inverse of $\lambda_{j_i}$ modulo $m$.
	By Lemma \ref{simple_criterion}, $(a_0,\cdots,a_{s-1},a)\in G_0(Q_{j_0},\cdots,Q_{j_{s-1}},Q_{j_s})$ if and only if, for all $0\leq k\leq s-1$, the following two conditions are both satisfied:
	\begin{align}
	&\sum_{i=0}^{s-1}\left\lfloor\frac{a_i-a\la\la_{j_i}}{m}\right\rfloor + \left\lfloor\frac{a-a\la\la_{j_s}}{m}\right\rfloor+\sum_{i=s+1}^{r}\left\lfloor\frac{-a\la\la_{j_i}}{m}\right\rfloor \leq -1, \label{condition1}\\
	&\sum_{i=0}^{s-1}\left\lfloor\frac{a_i-a_k\sigma_{j_k}\la_{j_i}}{m}\right\rfloor + \left\lfloor\frac{a-a_k\sigma_{j_k}\la_{j_s}}{m}\right\rfloor+\sum_{i=s+1}^{r}\left\lfloor\frac{-a_k\sigma_{j_k}\la_{j_i}}{m}\right\rfloor \leq -1. \label{condition2}
	\end{align}
	
	For each $0\leq k\leq s-1$, note that $\floor*{\frac{a-a_k\sigma_{j_k}\lambda_{j_s}}{m}}\neq \floor*{\frac{a-1-a_k\sigma_{j_k}\lambda_{j_s}}{m}}$ if and only if $a\lambda\equiv a_k\sigma_{j_k}\pmod m$.
	Suppose that $a\lambda\equiv a_k\sigma_{j_k}\pmod m$. According to Lemma \ref{sum_mod}, we have
	\begin{align*}
		&~\sum_{i=0}^{s-1}\left\lfloor\frac{a_i-a_k\sigma_{j_k}\la_{j_i}}{m}\right\rfloor + \left\lfloor\frac{a-a_k\sigma_{j_k}\la_{j_s}}{m}\right\rfloor+\sum_{i=s+1}^{r}\left\lfloor\frac{-a_k\sigma_{j_k}\la_{j_i}}{m}\right\rfloor\\
		=&~\sum_{i=0}^{s-1}\left\lfloor\frac{a_i-a\la\la_{j_i}}{m}\right\rfloor + \left\lfloor\frac{a-a\la\la_{j_s}}{m}\right\rfloor+\sum_{i=s+1}^{r}\left\lfloor\frac{-a\la\la_{j_i}}{m}\right\rfloor.
	\end{align*}
	Then the condition (\ref{condition2}) is equivalent to the condition (\ref{condition1}).
	Suppose that $a\lambda\not\equiv a_k\sigma_{j_k}\pmod m$. Then we have
	\begin{align*}
	~&\sum_{i=0}^{s-1}\left\lfloor\frac{a_i-a_k\sigma_{j_k}\la_{j_i}}{m}\right\rfloor + \left\lfloor\frac{a-a_k\sigma_{j_k}\la_{j_s}}{m}\right\rfloor+\sum_{i=s+1}^{r}\left\lfloor\frac{-a_k\sigma_{j_k}\la_{j_i}}{m}\right\rfloor\\
	=~& \sum_{i=0}^{s-1}\left\lfloor\frac{a_i-a_k\sigma_{j_k}\la_{j_i}}{m}\right\rfloor + \left\lfloor\frac{a-1-a_k\sigma_{j_k}\la_{j_s}}{m}\right\rfloor+\sum_{i=s+1}^{r}\left\lfloor\frac{-a_k\sigma_{j_k}\la_{j_i}}{m}\right\rfloor.
	\end{align*}
	If $a=1$, since$(a_0,\cdots,a_{s-1})\in G_0(Q_{j_0}$, $\cdots$, $Q_{j_{s-1}})$, by Lemma \ref{simple_criterion}, the condition (\ref{condition2}) is satisfied.
	If $a\geq2$, since $(a_0,\cdots,a_{s-1},a+1)\in G_0(Q_{j_0},\cdots, Q_{j_{s-1}}, Q_{j_s})$, by Lemma \ref{simple_criterion}, the condition (\ref{condition2}) is satisfied.
	
	Therefore, we obtain that $(a_0,\cdots,a_{s-1},a)\in G_0(Q_{j_0},\cdots,Q_{j_{s-1}},Q_{j_s})$ if and only if the condition (\ref{condition1}) is satisfied.
	This concludes the proof.
\end{proof}

\section{Pure gaps from special Kummer extensions}\label{sec4}
Throughout this section, we consider a special Kummer extension $F = K(x,y)/K(x)$ defined by the affine equation
\begin{equation}\label{equation2}
	y^m = f(x)^\la = \prod_{i=1}^r(x-\alpha_i)^\la,
\end{equation}
where $\lambda\in \mathbb{Z}\setminus\{0\}$, $\gcd(r\la,m) = 1$, and $\alpha_1,\cdots,\alpha_r\in K$ are pairwise distinct elements. The function field $F/K$ has genus $g = (r-1)(m-1)/2$.

For each $1\leq i\leq r$, let $Q_i$ and $Q_\infty$ denote the places in $\P(F)$ corresponding to the zero of $x-\a_i$ and the pole of $x$, respectively.
It is obvious that $ Q_1,\cdots, Q_r$ and $Q_{\infty}$ all are rational and totally ramified.
An explicit description of $G_0(Q_1,\cdots,Q_s)$ has been provided by Castellanos et al. in \cite[Theorem 5.12]{castellanosSetPure2024}.
However, an explicit description of $G_0(Q_\infty, Q_1,\cdots,Q_s)$ has not been provided in the literature thus far.

In this section, applying Theorem \ref{bottom_set} and Theorem \ref{find_another_pure_gap} in the previous section, we aim to establish explicit descriptions of $G_0(Q_1,\cdots,Q_s)$ and $G_0(Q_\infty,Q_1,\cdots,Q_s)$ for $2\leq s\leq r$.
Our description of $G_0(Q_1,\cdots,Q_s)$ aligns with the one presented in \cite[Theorem 5.12]{castellanosSetPure2024}.
In the case $mv = r+1$ for some $v\geq 1$, our description of $G_0(Q_\infty,Q_1,\cdots,Q_s)$ bears a resemblance to that of $G_0(Q_1,\cdots,Q_{s})$.
Our descriptions can be formed as the union of sets of consecutive pure gaps, enabling us to easily obtain families of consecutive pure gaps.
Specifically, our main results are as follows.
\begin{proposition}\label{precise_puregap1}
	Suppose that $2\leq s\leq r$.
	If $2\leq s\leq r-\floor*{\frac{r}{m}}-1$, for each $0\leq k\leq r-\floor*{\frac{r}{m}}-1-s$, let $t_{i,k} = m-\ceil*{\frac{m(k+i)}{r}}$ for $1\leq i\leq s$.
	Then
	\begin{align*}
		G_0(Q_1,\cdots,Q_s) = \bigcup\limits_{k=0}^{r-\floor*{\frac{r}{m}}-1-s} \Big\{(mk_1+j_1,\cdots,mk_s+j_s)~\Big|~ k_i\geq 0 \text{ for } 1\leq i\leq s,\\
		\sum_{i=1}^s k_i = k,~ 1\leq j_{\sigma(i)}\leq t_{i,k} \text{ for } 1\leq i\leq s,~\sigma\in S_s\Big\},
	\end{align*}
	where  $S_s$ is the set of permutations of the set $\{1,\cdots,s\}$.

	If $r-\floor*{\frac{r}{m}}\leq s \leq r$, then $G_0(Q_1,\cdots,Q_s) = \varnothing$.
\end{proposition}

\begin{proposition}\label{precise_puregap2}
	Suppose that $mv = r+1$ for some $v \geq 1$ and $1\leq s\leq r$.
	If $1\leq s\leq r-v-1$, for each $0\leq k\leq r-v-1-s$, let $t_{i,k} = m-\ceil*{\frac{m(k+i+1)}{r}}$ for $0\leq i\leq s$.
	Then
	\begin{align*}
		&G_0(Q_\infty,Q_1,\cdots,Q_s) = \bigcup\limits_{k=0}^{r-v-1-s} \Big\{(mk_0+j_0,mk_1+j_1,\cdots,mk_s+j_s)~\Big|\\
		k_i\geq 0 &\text{ for } 0\leq i\leq s,~\sum_{i=0}^s k_i = k,~ 1\leq j_{\sigma(i)}\leq t_{i,k} \text{ for } 0\leq i\leq s,~\sigma\in S_{s+1}\Big\},
	\end{align*}
	where  $S_{s+1}$ is the set of permutations of the set $\{0,1,\cdots,s\}$.

	If $r-v\leq s \leq r$, then $G_0(Q_\infty,Q_1,\cdots,Q_s) = \varnothing$.
\end{proposition}

\begin{proposition}\label{precise_puregap3}
	Suppose that $1\leq s\leq r$.
	If $1\leq s\leq r-\floor*{\frac{r}{m}}-1$,
	for each $1\leq j_0\leq m-1-\floor*{\frac{m}{r}}$ and $1\leq k\leq r-1$, let 
	$$t_{j_0,i,k} = \left\{\begin{array}{cc}
		m-\ceil*{\frac{(k+i)m}{r}}+j_0,~&\text{if}~ k+i \neq r,\\
		m-\ceil*{\frac{(k+i)m}{r}}+j_0-1,~&\text{if}~ k+i = r,\\
	\end{array}\right.$$ for all $1\leq i\leq s$, and $d_{j_0} = r-\floor*{\frac{(1-j_0)r}{m}}-s-1$.
	Let $c = m-1-\floor*{\frac{m}{r}}$. Then
	\begin{align*}
		&G_0(Q_\infty,Q_1,\cdots,Q_s) = \bigcup_{j_0 = 1}^{c}\bigcup_{k=\ceil*{\frac{rj_0}{m}}}^{d_{j_0}}\Big\{(mk_0-rj_0,mk_1+j_1,\cdots,mk_s+j_s)~\Big|~k_0\geq\ceil*{\frac{rj_0}{m}}, \\
		&~k_i\geq 0 \text{ for } 1\leq i\leq s,~ \sum_{i=0}^s k_i = k,~ 1\leq j_{\sigma(i)}\leq t_{j_0,i,k} \text{ for } 1\leq i\leq s,~\sigma \in S_s\Big\},
	\end{align*}
	where  $S_s$ is the set of permutations of the set $\{1,\cdots,s\}$.

	If $r-\floor*{\frac{r}{m}}\leq s\leq r$, then $G_0(Q_\infty,Q_1,\cdots,Q_s) = \varnothing$.
\end{proposition}

We will divide into three subsections to prove each of the above three propositions. Initially, we present a lemma.
\begin{lemma}\label{equivalent}
	Let $k$ be an integer such that $1\leq k \leq 2r-1$ and let $j$ be an integer such that $1-m\leq j\leq m-1$.
	Then
	$$\floor*{\frac{rj}{m}}\leq r-k-1 \text{ if and only if } j\leq \left\{
		\begin{array}{cc}
			m-\ceil*{\frac{km}{r}},~&\text{if}~ k \neq r,\\
			m-\ceil*{\frac{km}{r}}-1,~&\text{if}~ k = r.
		\end{array}\right.$$
	Furthermore suppose that $mv = r +1$ for some $v\geq 1$, $1\leq k\leq r-1$, and $1\leq j\leq m-1$.
	Then
	$$j\leq m-\ceil*{\frac{km}{r}} \text{ if and only if } vj \leq r -k.$$
\end{lemma}
\begin{proof}
	First, assume that $k\neq r$.
	Suppose that $j\neq 0$. Since $1-m\leq j\leq m-1$ and $\gcd(m,r) = 1$, it follows that $m\nmid rj$.
	If $\floor*{\frac{rj}{m}}\leq r-k-1$, then $\frac{rj}{m}< r-k$.
	This implies that $j<m-\frac{km}{r}$, and thus $j\leq \floor*{m-\frac{km}{r}} = m-\ceil*{\frac{km}{r}}$.
	If $j\leq m-\ceil*{\frac{km}{r}}$, then $j\leq m-\ceil*{\frac{km}{r}} = m + \floor*{\frac{-km}{r}} = \floor*{\frac{(r-k)m}{r}}$.
	This implies that $rj\leq \floor*{\frac{(r-k)m}{r}}r\leq (r-k)m$, and thus $\frac{rj}{m}\leq r-k$.
	Since $m\nmid rj$, it follows that $\frac{rj}{m}< r-k$, and hence $\floor*{\frac{rj}{m}}\leq r-k-1$.
	
	Now, suppose that $j = 0$. We need to prove that $0\leq r-k-1$ if and only if $0\leq m-\ceil*{\frac{km}{r}}$.
	If $k\leq r-1$, then $m-\ceil*{\frac{km}{r}}\geq m-\ceil*{\frac{(r-1)m}{r}}\geq 0$.
	If $0\leq m-\ceil*{\frac{km}{r}} = \floor*{\frac{m(r-k)}{r}}$, given $k\neq r$ and $(m,r) = 1$, it follows that $r\nmid m(r-k)$.
	Thus $\frac{m(r-k)}{r}>0$ and then $r-k>0$. This yields that $r-k-1\geq 0$.
	
	Next, assume that $k = r$. It is obvious that $\floor*{\frac{rj}{m}} \leq -1 = r-k-1$ if and only if $j\leq -1 = m-\ceil*{\frac{km}{r}}-1$. 
	
	Furthermore suppose that $mv = r +1$ for some $v\geq 1$, $1\leq k\leq r-1$, and $1\leq j\leq m-1$.
	Then $\floor*{\frac{rj}{m}} = vj + \floor*{\frac{-j}{m}} = vj-1$ and $k\neq r$.
	It follows that $j\leq m-\ceil*{\frac{km}{r}}$ if and only if $vj \leq r -k$.
\end{proof}

For each $1\leq i\leq r$, since $g = \# G(Q_i)$ , by Lemma \ref{special_gap}, we obtain that 
\begin{equation}\label{number_of_gap}
\frac{(r-1)(m-1)}{2} =  \sum_{j=1}^{m-1-\floor*{\frac{m}{r}}}\left(r-1-\floor*{\frac{rj}{m}}\right).
\end{equation} 
Using this equation and Lemma \ref{equivalent}, we present a novel characterization of $G(Q_i)$ for $1\leq i\leq r$ and $G(Q_\infty)$ as follows.
\begin{proposition}\label{novel_gap}
	For each $1\leq i\leq m$,
	$$G(Q_i) = \left\{ mk+j ~\Big| ~0\leq k\leq r-2-\floor*{\frac{r}{m}},~1\leq j\leq m-\ceil*{\frac{(k+1)m}{r}} \right\},$$
	Furthermore, in the case where $mv = r +1$ for some $v\geq 1$,
	$$G(Q_\infty) = \left\{ mk+j ~\Big| ~0\leq k\leq r-v-1,~1\leq j\leq m-\ceil*{\frac{(k+1)m}{r}} \right\}.$$
	In the case where $m = ur +1$ for some $u\geq 1$,
	$$G(Q_\infty) = \left\{ mk-rj ~| ~1\leq k\leq r-1,~1\leq j\leq ku\right\}.$$
\end{proposition}
\begin{proof}
	Let $G_1 = \left\{ mk+j ~|~ ~0\leq k\leq r-2-\floor*{\frac{r}{m}},~1\leq j\leq m-\ceil*{\frac{(k+1)m}{r}} \right\}$.
	It is sufficient to show that $G(Q_1) = G_1$.
	For $0\leq k\leq r-2-\floor*{\frac{r}{m}}$ and $1\leq j\leq m-1$,
	\begin{align*}
		\floor*{\frac{(mk+j)r}{m}}-(r-1)\ceil*{\frac{mk+j}{m}}& = ~kr+\floor*{\frac{rj}{m}}-(r-1)(k+1) = \floor*{\frac{rj}{m}}-r+k+1
	\end{align*}
	By Proposition \ref{gap_set}, $mk+j\in G(Q_1)$ if and only if $\floor*{\frac{rj}{m}} \leq r-(k+1)-1$.
	It follows from Lemma \ref{equivalent} that $\floor*{\frac{rj}{m}} \leq r-(k+1)-1$ if and only if  $j\leq m-\ceil*{\frac{(k+1)m}{r}}$.
	Then $G_1\subseteq G(Q_1)$ by Proposition \ref{gap_set}. On the other hand,
	by (\ref{number_of_gap}), we have
	\begin{align*}
		\#G_1 &= \sum_{k=0}^{r-2-\floor*{\frac{r}{m}}}\left(m-\ceil*{\frac{(k+1)m}{r}}\right) \\
		&= \sum_{k=1}^{r-1-\floor*{\frac{r}{m}}}\left(m-1-\floor*{\frac{km}{r}}\right) = \frac{(m-1)(r-1)}{2} = g.
	\end{align*}
	Thus we conclude that $G_1 = G(Q_1)$. 
	
	Furthermore, in the case where $mv = r +1$ for some $v\geq 1$,
	by Remark \ref{special_infinity}, we have $G(Q_\infty) = G(Q_1)$.
	In the case where $m = ur+1$ for some $u\geq 1$, let $G_\infty = \left\{ mk-rj ~|~ 1\leq k\leq r-1,1\leq j\leq ku \right\}$.
	In this case, $-u$ is the inverse of $r$ modulo $m$. For $1\leq k\leq r-1$ and $1\leq j\leq ku$, by Lemma \ref{sum_mod}, we have
	\begin{align*}
		&~ r\floor*{\frac{-(mk-rj)u}{m}} + \floor*{\frac{(mk-rj)ur}{m}}  +\ceil*{\frac{mk-rj}{m}}\\
		=&~ -kur + r\floor*{\frac{-j}{m}} + kur + \floor*{\frac{rj}{m}} + k + \ceil*{\frac{-rj}{m}} =  k - r \leq -1.
	\end{align*}
	Then  by Proposition \ref{gap_set}, we have $G_\infty\subseteq G(Q_\infty)$. On the other hand, we have
	$$\#G_\infty = \sum_{k=1}^{r-1} ku = \frac{(r-1)ur}{2} = \frac{(r-1)(m-1)}{2} = g.$$
	Thus we conclude that $G_\infty = G(Q_\infty)$.
\end{proof}

\subsection{\texorpdfstring{Pure gaps at $Q_1,\cdots,Q_s$}{Pure gaps at Q1, ..., Qs}}
Using the description of $G(Q_i)$ for $1\leq i \leq r$ presented in Proposition \ref{novel_gap}, we obtain a family of consecutive pure gaps as follows.
\begin{proposition}\label{precisefamily1}
		Suppose that $2\leq s\leq r-\floor*{\frac{r}{m}}-1$ and $0\leq k\leq r-\floor*{\frac{r}{m}}-1-s$. 
		Let $S_s$ be the set of permutations of the set $\{1,\cdots,s\}$.
		Let $t_i = m -\ceil*{\frac{(k+i)m}{r}}$ for $1\leq i\leq s$.
		Then for all $\sigma\in S_s$,
		$$\{(mk+j_1,j_2,\cdots,j_s)~|~ 1\leq j_{i}\leq t_i \text{ for } 1\leq i\leq s\}\subseteq G_0(Q_{\sigma(1)},Q_{\sigma(2)},\cdots, Q_{\sigma(s)}).$$
\end{proposition}
\begin{proof}
	For each $1\leq i\leq s$, by Lemma \ref{equivalent}, we have
	\begin{equation}\label{equivalent1}
		j_{i}\leq m-\ceil*{\frac{(k+i)m}{r}} \text{ if and only if } \floor*{\frac{rj_{i}}{m}}\leq r-(k+i)-1.
	\end{equation} 
	We set $G = \{(mk+j_1,j_2,\cdots,j_s)~|~ 1\leq j_{i}\leq t_{i} \text{ for } 1\leq i\leq s\}$.
	For any $(mk+j_1,j_2,\cdots,j_s)\in G$ and $\sigma\in S$, by Proposition \ref{novel_gap}, 
	it is obvious that $mk+j_1 \in G(Q_{\sigma(1)})$.
	We will show that $(mk+j_{1},j_2,\cdots,j_{l})\in G_0(Q_{\sigma(1)},Q_{\sigma(2)}\cdots,Q_{\sigma(l)})$ for $2\leq l\leq s$ by induction.
	Let $\lambda^\prime\in \Z$ be the inverse of $\lambda$ modulo $m$.
	For each $1\leq j_{2}\leq t_2$, by Lemma \ref{sum_mod}, we have
	\begin{align*}
		&~\left\lfloor\frac{mk+j_{1}-j_{2}\la\la^\prime}{m}\right\rfloor+\left\lfloor\frac{j_{2}-j_{2}\la\la^\prime}{m}\right\rfloor +\sum_{i=3}^{r}\left\lfloor\frac{-j_{2}\la\la^\prime}{m}\right\rfloor + \left\lfloor\frac{j_{2}r\la\la^\prime}{m}\right\rfloor\\
		=&~k+\sum_{i=1}^2\left\lfloor\frac{j_{i}-j_{2}}{m}\right\rfloor +\sum_{i=3}^{r}\left\lfloor\frac{-j_{2}}{m}\right\rfloor + \left\lfloor\frac{rj_{2}}{m}\right\rfloor \leq \left\lfloor\frac{rj_{2}}{m}\right\rfloor -r+2+k.
	\end{align*}
	It follows from (\ref{equivalent1}) that $\floor*{\frac{rj_{2}}{m}}\leq r-(k+2)-1$, and then $\left\lfloor\frac{rj_{2}}{m}\right\rfloor -r+2+k\leq -1$.
	Therefore,  by Theorem \ref{find_another_pure_gap}, we have $(mk+j_{1},j_{2})\in G_0(Q_{\sigma(1)},Q_{\sigma(2)})$.
	
	Assume that $(mk+j_1,j_2,\cdots,j_{l-1})\in G_0(Q_{\sigma(1)},Q_{\sigma(2)},\cdots,Q_{\sigma(l-1)})$ for some $3\leq l\leq s$.
	For each $1\leq j_{l}\leq t_l$, by Lemma \ref{sum_mod}, we have
	\begin{align*}
		&~\left\lfloor\frac{mk+j_{1}-j_{l}\la\la^\prime}{m}\right\rfloor+\sum_{i=2}^{l}\left\lfloor\frac{j_{i}-j_{l}\la\la^\prime}{m}\right\rfloor + \sum_{i=l+1}^{r}\left\lfloor\frac{-j_{l}\la\la^\prime}{m}\right\rfloor + \left\lfloor\frac{j_{l}r\la\la^\prime}{m}\right\rfloor\\
	  =&~ k+\sum_{i=1}^{l}\left\lfloor\frac{j_{i}-j_{l}}{m}\right\rfloor +\sum_{i=l+1}^{r}\left\lfloor\frac{-j_{l}}{m}\right\rfloor + \left\lfloor\frac{rj_{l}}{m}\right\rfloor \leq \left\lfloor\frac{rj_{l}}{m}\right\rfloor-r+l+k.
	\end{align*}
	It follows from (\ref{equivalent1}) that $\floor*{\frac{rj_{l}}{m}}\leq r-(k+l)-1$, and then $\left\lfloor\frac{rj_{l}}{m}\right\rfloor -r+l+k \leq -1$.
	Therefore, by Theorem \ref{find_another_pure_gap}, we have $(mk+j_1,j_2,\cdots,j_{l})\in G_0(Q_{\sigma(1)},Q_{\sigma(2)},\cdots,Q_{\sigma(l)})$.
	By induction, we obtain that $(mk+j_1,j_2,\cdots,j_{s})\in G_0(Q_{\sigma(1)},Q_{\sigma(2)},\cdots,Q_{\sigma(s)})$.
	Thus we have $G\subseteq G_0(Q_{\sigma(1)},Q_{\sigma(2)},\cdots,Q_{\sigma(s)})$ for all $\sigma\in S_s$.	
\end{proof}

	In a recent study by Castellanos et al.\cite[Theorem 5.12]{castellanosSetPure2024}, the authors determine a simple explicit description of set of pure gaps at several totally ramified places distinct to the infinite place $Q_\infty$.
	Here, we present an alternative proof by applying our previous results.
	\begin{proof}[Proof of Proposition \ref{precise_puregap1}]
	If $2\leq s\leq r-\floor*{\frac{r}{m}}-1$, for each $0\leq k\leq r-\floor*{\frac{r}{m}}-1-s$, we set $G_k = \{(mk_1+j_1,\cdots,mk_s+j_s)~|~  k_i\geq 0 \text{ for } 1\leq i\leq s,~\sum_{i=1}^s k_i=k,~1\leq j_{\sigma(i)}\leq t_{i,k} \text{ for } 1\leq i\leq s,~\sigma\in S_s\}$
	and  $\hat{G}_k = \{(mk+j_1,j_2,\cdots,j_s)~|~1\leq j_{i}\leq t_{i,k} \text{ for } 1\leq i\leq s\}$.
	According to Proposition \ref{precisefamily1}, we have $\hat{G}_k \subseteq G_0(Q_{\sigma(1)},\cdots,Q_{\sigma(s)})$ for all $\sigma\in S_s$.
	It follows from Theorem \ref{bottom_set} that $G_k\subseteq G_0(Q_1,\cdots,Q_s)$.
	
	For each $1\leq i\leq s$ and $0\leq k\leq r-\floor*{\frac{r}{m}}-1-s$, by Lemma \ref{equivalent}, we have
	\begin{equation}\label{equivalent2}
		j_{\sigma(i)}\leq m-\ceil*{\frac{(k+i)m}{r}} \text{ if and only if } \floor*{\frac{rj_{\sigma(i)}}{m}}\leq r-(k+i)-1.
	\end{equation} 
	Next we will show that  $$G_0(Q_1,\cdots,Q_s)\subseteq \bigcup\limits_{k=0}^{r-\floor*{\frac{r}{m}}-1-s} G_k.$$
	
	For any $(a_1,\cdots,a_s)\in G_0(Q_1,\cdots,Q_s)$,
	by Propositions \ref{subpuregap} and \ref{novel_gap}, we can write $a_i = mk_i+j_i$ with $0\leq k_i\leq r-2-\floor*{\frac{r}{m}}$ and $1\leq j_i\leq m-\ceil*{\frac{(k_i+1)m}{r}}$ for $1\leq i\leq s$.
	There exists some $\sigma\in S_s$ such that $j_{\sigma(1)}\geq j_{\sigma(2)}\geq \cdots\geq j_{\sigma(s)}$.
	Let $\lambda^\prime\in \Z$ be the inverse of $\lambda$ modulo $m$.
	Then for each $1\leq l\leq s$, by Lemma \ref{sum_mod}, we have 
	\begin{align*}
		&~\sum_{i=1}^{s}\left\lfloor\frac{a_{\sigma(i)}-a_{\sigma(l)}\la\la^\prime}{m}\right\rfloor+\sum_{i=s+1}^{r}\left\lfloor\frac{-a_{\sigma(l)}\la\la^\prime}{m}\right\rfloor + \left\lfloor\frac{a_{\sigma(l)} r\la\la^\prime}{m}\right\rfloor\\
	  = &~\sum_{i=1}^{s}k_{\sigma(i)}+\sum_{i=1}^{s}\left\lfloor\frac{j_{\sigma(i)}-j_{\sigma(l)}}{m}\right\rfloor+\sum_{i=s+1}^{r}\left\lfloor\frac{-j_{\sigma(l)}}{m}\right\rfloor + \left\lfloor\frac{rj_{\sigma(l)}}{m}\right\rfloor\\
	  = &~\sum_{i=1}^{s}k_{i}+\sum_{i=l+1}^{s}\left\lfloor\frac{j_{\sigma(i)}-j_{\sigma(l)}}{m}\right\rfloor+\left\lfloor\frac{rj_{\sigma(l)}}{m}\right\rfloor -r+s.
	\end{align*}
	Let $k = \sum_{i=1}^{s}k_i$. When $l=s$, by Lemma \ref{simple_criterion}, we must have $k+\floor*{\frac{rj_{\sigma(s)}}{m}}-r+s\leq -1$,
	which is equivalent to $s\leq r-\floor*{\frac{rj_{\sigma(s)}}{m}}-k-1$.
	Since $j_{\sigma(s)}\geq 1$, we must have that $k\leq r-\floor*{\frac{r}{m}}-1-s$.
	This implies that $s\leq r-\floor*{\frac{r}{m}}-1$ since $k\geq 0$. Thus we have that $Q_0(Q_1,\cdots,Q_s) = \varnothing$ for $r-\floor*{\frac{r}{m}}\leq s\leq r$.
	
	Again by Lemma \ref{simple_criterion}, for $1\leq l\leq s$,  we have $k+\sum_{i=l+1}^{s}\left\lfloor\frac{j_{\sigma(i)}-j_{\sigma(l)}}{m}\right\rfloor+\left\lfloor\frac{rj_{\sigma(l)}}{m}\right\rfloor -r+s \leq -1,$
	which is equivalent to 
	$$\left\lfloor\frac{rj_{\sigma(l)}}{m}\right\rfloor\leq r-k-s-1-\sum_{i=l+1}^{s}\left\lfloor\frac{j_{\sigma(i)}-j_{\sigma(l)}}{m}\right\rfloor.$$
	Since $\sum_{i=l+1}^{s}\left\lfloor\frac{j_{\sigma(i)}-j_{\sigma(l)}}{m}\right\rfloor\geq -(s-l)$, we have 
	$$\left\lfloor\frac{rj_{\sigma(l)}}{m}\right\rfloor\leq r-k-s-1+(s-l) = r-(k+l)-1.$$
	It follows from (\ref{equivalent1}) that $(a_1,\cdots,a_s)\in G_k$, and then 
	$$G_0(Q_1,\cdots,Q_s)\subseteq \bigcup\limits_{k=0}^{r-\floor*{\frac{r}{m}}-1-s} G_k.$$
	The proof is complete.
	\end{proof}

	\begin{corollary}\label{precise_bottom1}
	Suppose that $2\leq s\leq r$.
	If $2\leq s\leq r-\floor*{\frac{r}{m}}-1$, let $t_{i} = m-\ceil*{\frac{im}{r}}$ for $1\leq i\leq s$.
	Then
	$$\overline{G}_0(Q_1,\cdots,Q_s) = \{(j_1,\cdots,j_s) ~|~ 1\leq j_{\sigma(i)}\leq t_{i} \text{ for } 1\leq i\leq s,~\sigma\in S_s\},$$
	where $S_s$ is the set of permutations of the set $\{1,\cdots,s\}$.

	If $r-\floor*{\frac{r}{m}}\leq s \leq r$, then $\overline{G}_0(Q_1,\cdots,Q_s) = \varnothing$.
\end{corollary}
\begin{proof}
	If $2\leq s\leq r-\floor*{\frac{r}{m}}-1$, for each $0\leq k\leq r-\floor*{\frac{r}{m}}-1-s$, let $t_{i,k} = m-\ceil*{\frac{m(k+i)}{r}}$ for $1\leq i\leq s$.
	According to Proposition \ref{precise_puregap1}, we have 
	\begin{align*}
		\overline{G}_0(Q_1,\cdots,Q_s) &= \bigcup\limits_{k=0}^{r-\floor*{\frac{r}{m}}-1-s} \{(j_1,\cdots,j_s)~|~1\leq j_{\sigma(i)}\leq t_{i,k} \text{ for } 1\leq i\leq s, ~\sigma\in S_s\}\\
		&= \{(j_1,\cdots,j_s)~|~1\leq j_{\sigma(i)}\leq t_{i,0} \text{ for } 1\leq i\leq s, ~\sigma\in S_s\}\\
		&= \{(j_1,\cdots,j_s) ~|~ 1\leq j_{\sigma(i)}\leq t_{i} \text{ for } 1\leq i\leq s,~\sigma\in S_s\},
	\end{align*}
	where $S_s$ is the set of permutations of the set $\{1,\cdots,s\}$.
	
	If $r-\floor*{\frac{r}{m}}\leq s \leq r$, then $\overline{G}_0(Q_1,\cdots,Q_s) = \varnothing$ since $G_0(Q_1,\cdots,Q_s) = \varnothing$.
\end{proof}

\subsection{\texorpdfstring{Pure gaps at $Q_\infty, Q_1,\cdots,Q_s$ for $mv=r+1$}{Pure gaps at Q0, Q1, ..., Qs for mv=r+1}}

In the case where $mv = r+1$ for some $v \geq 1$, by Proposition \ref{novel_gap}, we have $G(Q_\infty) = G(Q_i)$ for all $1\leq i\leq r$.
In this subsection, we will prove Proposition \ref{precise_puregap2}, which shows that the description of $G_0(Q_\infty,Q_1,\cdots,Q_s)$ bears resemblance to that of $G_0(Q_1,\cdots,Q_{s})$ for $1\leq s\leq r$.
Firstly, we present a family of consecutive pure gaps as follows.
\begin{proposition}\label{precisefamily2}
	Suppose that $mv = r+1$ for some $v \geq 1$, $1\leq s\leq r-v-1$, and $0\leq k\leq r-v-1-s$.
	Let $S_{s+1}$ be the set of permutations of the set $\{0,1,\cdots,s\}$ and $Q_0 = Q_\infty$. 
	Let $t_i = m-\ceil*{\frac{m(k+i+1)}{r}}$ for $0\leq i\leq s$.
	Then for all $\sigma\in S_{s+1}$,
	$$\{(mk+j_0,j_1,\cdots,j_s)~|~ 1\leq j_{i}\leq t_i \text{ for } 0\leq i\leq s\}\subseteq G_0(Q_{\sigma(0)},Q_{\sigma(1)},\cdots, Q_{\sigma(s)}).$$
\end{proposition}
\begin{proof}
	For each $0\leq i\leq s$, by Lemma \ref{equivalent}, we have
	\begin{equation}\label{equivalent3}
		j_{i}\leq m-\ceil*{\frac{(k+i+1)m}{r}} \text{ if and only if } vj_i\leq r-(k+i+1).
	\end{equation} 
	We set $G =  \{(mk+j_0,j_1,\cdots,j_s)~|~ 1\leq j_{i}\leq t_i \text{ for } 0\leq i\leq s\}$.
	For any $(mk+j_0,j_1,\cdots,j_s)\in G$ and $\sigma\in S_{s+1}$, by Proposition \ref{novel_gap}, it is obvious that $mk+j_0\in G(Q_{\sigma(0)})$.
	We will show that $(mk+j_0,j_1,\cdots,j_l)\in G_0(Q_{\sigma(0)},Q_{\sigma(1)},\cdots,Q_{\sigma(l)})$ for $1\leq l\leq s$ by induction.
	Let $\lambda^\prime\in \Z$ be the inverse of $\lambda$ modulo $m$. Then $\lambda'$ is also the inverse of $-r\lambda$ modulo $m$.
	Let $\lambda_ 0  = -r\lambda$ and $\lambda_i = \lambda$ for $1\leq i\leq s$.
	For each $1\leq j_1\leq t_1$, by Lemma \ref{sum_mod} and $r = mv-1$, we have
	\begin{align*}
		&~\left\lfloor\frac{mk+j_0-j_1\la^\prime\la_{\sigma(0)}}{m}\right\rfloor +\left\lfloor\frac{j_1-j_1\la^\prime\la_{\sigma(1)}}{m}\right\rfloor +\sum_{i=2}^{r}\left\lfloor\frac{-j_1\la^\prime\la_{\sigma(i)}}{m}\right\rfloor\\
		 = &~k + vj_{1} + \left\lfloor\frac{j_0-j_1}{m}\right\rfloor +\sum_{i=2}^{r}\left\lfloor\frac{-j_1}{m}\right\rfloor \leq k + vj_1 - r+1 .
	\end{align*}
	It follows from (\ref{equivalent3}) that $vj_1\leq r-k-1-1$, and then $k + vj_1 - r+1\leq -1$.
	Therefore, by Theorem \ref{find_another_pure_gap}, we have $(mk+j_0,j_1)\in G_0(Q_{\sigma(0)},Q_{\sigma(1)})$.
	
	Assume that $(mk+j_0,j_1,\cdots,j_{l-1})\in G_0(Q_{\sigma(0)},Q_{\sigma(1)},\cdots,Q_{\sigma(l-1)})$ for some $2\leq l\leq s$.
	For each $1\leq j_l\leq t_l$, by Lemma \ref{sum_mod}, we have
	\begin{align*}
		&~\left\lfloor\frac{mk+j_0-j_l\la^\prime\la_{\sigma(0)}}{m}\right\rfloor + \sum_{i=1}^{l}\left\lfloor\frac{j_i-j_l\la_{\sigma(i)}\la^\prime}{m}\right\rfloor + \sum_{i=l+1}^{r}\left\lfloor\frac{-j_l\la^\prime \la_{\sigma(i)}}{m}\right\rfloor\\
	  = &~k+ vj_l + \sum_{i=0}^{l}\left\lfloor\frac{j_{i}-j_{l}}{m}\right\rfloor +\sum_{i=l+1}^{r}\left\lfloor\frac{-j_{l}}{m}\right\rfloor \leq k + vj_l -r+l.
	\end{align*}
	It follows from (\ref{equivalent3}) that $vj_l\leq r-k-l-1$, and then $k + vj_l -r+l\leq -1$.
	Therefore, by Theorem \ref{find_another_pure_gap}, we have $(mk+j_0,j_1,\cdots,j_l)\in G_0(Q_{\sigma(0)},Q_{\sigma(1)},\cdots,Q_{\sigma(l)})$.
	By induction, we obtain that $(mk+j_0,j_1,\cdots,j_s)\in G_0(Q_{\sigma(0)},Q_{\sigma(1)},\cdots,Q_{\sigma(s)})$.
	Thus we have $ G\subseteq G_0(Q_{\sigma(0)},Q_{\sigma(1)},\cdots,Q_{\sigma(s)})$.
\end{proof}
	Based on the above proposition, we will finish the proof of Proposition \ref{precise_puregap2} in the following.
\begin{proof}[Proof of Proposition \ref{precise_puregap2}]
	If $1\leq s\leq r-v-1$, for each $0\leq k\leq r-v-1-s$, we set $G_k = \{(mk_0+j_0,mk_1+j_1,\cdots,mk_s+j_s)~|~
	k_i\geq 0 \text{ for } 0\leq i\leq s,~\sum_{i=0}^s k_i = k,~ 1\leq j_{\sigma(i)}\leq t_{i,k} \text{ for } 0\leq i\leq s,~\sigma\in S_{s+1}\}$
	and  $\hat{G}_k = \{(mk+j_0,j_1,\cdots,j_s)~|~1\leq j_i\leq t_{i,k} \text{ for } 0\leq i\leq s\}$.
	Let $Q_0 = Q_\infty$.
	According to Proposition \ref{precisefamily2}, we have $\hat{G}_k \subseteq G_0(Q_{\sigma(0)},Q_{\sigma(1)},\cdots,Q_{\sigma(s)})$ for all $\sigma\in S_{s+1}$.
	If follows from  Theorem \ref{bottom_set} that $G_k\subseteq G_0(Q_\infty,Q_1,\cdots,Q_s)$.

	For each $0\leq i\leq s$ and $0\leq k\leq r-v-1-s$, by Lemma \ref{equivalent}, we have
	\begin{equation}\label{equivalent4}
		j_{\sigma(i)}\leq m-\ceil*{\frac{(k+i+1)m}{r}} \text{ if and only if } vj_{\sigma(i)}\leq r-(k+i+1).
	\end{equation} 
	Next we will show that  $$G_0(Q_{\infty},Q_1,\cdots,Q_s)\subseteq \bigcup\limits_{k=0}^{r-v-1-s} G_k.$$
	For any $(a_0,a_1,\cdots,a_s)\in G_0(Q_\infty,Q_1,\cdots,Q_s)$, by Propositions \ref{subpuregap} and \ref{novel_gap}, we can write $a_i = mk_i+j_i$ with $0\leq k_i\leq r-2-\floor*{\frac{r}{m}}$ and $1\leq j_i\leq m-\ceil*{\frac{(k_i+1)m}{r}}$ for $0\leq i\leq s$.
	There exists some $\sigma\in S_{s+1}$ such that $j_{\sigma(0)}\geq j_{\sigma(1)}\geq \cdots\geq j_{\sigma(s)}$.
	Let $\lambda_ 0  = -r\lambda$ and $\lambda_i = \lambda$ for $1\leq i\leq s$.
	Let $\lambda^\prime\in \Z$ be the inverse of $\lambda$ modulo $m$.
	Then for each $1\leq l\leq s$, by Lemma \ref{sum_mod} and $r = mv-1$, we have 
	\begin{align*}
		&~\sum_{i=0}^{s}\left\lfloor\frac{a_{\sigma(i)}-a_{\sigma(l)}\la^\prime\la_{\sigma(i)}}{m}\right\rfloor + \sum_{i=s+1}^{r}\left\lfloor\frac{-a_{\sigma(l)}\la^\prime\la_{\sigma(i)}}{m}\right\rfloor\\
	  = &~\sum_{i=0}^{s}k_{\sigma(i)}+vj_{l} + \sum_{i=0}^{s}\left\lfloor\frac{j_{\sigma(i)}-j_{\sigma(l)}}{m}\right\rfloor+ \sum_{i=s+1}^{r}\left\lfloor\frac{-j_{\sigma(l)}}{m}\right\rfloor\\
	  = &~\sum_{i=0}^{s}k_{i}+vj_{l} +\sum_{i=l+1}^{s}\left\lfloor\frac{j_{\sigma(i)}-j_{\sigma(w)}}{m}\right\rfloor -r+s.
	\end{align*}
	Let $k = \sum_{i=0}^{s}k_{i}$. When $l = s$, by Lemma \ref{simple_criterion}, we must have $k+vj_{\sigma(s)}-r+s\leq -1$,
	which is equivalent to $s\leq r-k-vj_{\sigma(s)}-1$.
	Since $j_{\sigma(s)}\geq 1$, we must have that $k\leq r-v-1-s$. This implies that $s\leq r-v-1$ since $k\geq 0$.
	Thus we have that $G_0(Q_\infty,Q_1,\cdots,Q_s) = \varnothing$ for $r-v\leq s\leq r$.

	Again by Lemma \ref{simple_criterion}, for $0\leq l\leq s$, we have $k+vj_{\sigma(l)}+\sum_{i=l+1}^{s}\left\lfloor\frac{j_{\sigma(i)}-j_{\sigma(l)}}{m}\right\rfloor-r+s\leq -1$,
	which is equivalent to 
	$$rj_{\sigma(l)}\leq r-k-s-1-\sum_{i=l+1}^{s}\left\lfloor\frac{j_{\sigma(i)}-j_{\sigma(l)}}{m}\right\rfloor.$$
	Since $\sum_{i=l+1}^{s}\left\lfloor\frac{j_{\sigma(i)}-j_{\sigma(l)}}{m}\right\rfloor\geq -(s-l)$, we have 
	$$vj_{\sigma(l)}\leq r-k-s-1+(s-l) = r-k-l-1.$$
	It follows from  (\ref{equivalent4}) that $(a_0,a_1,\cdots,a_s)\in G_k$, and then 
	$$G_0(Q_{\infty},Q_1,\cdots,Q_s)\subseteq \bigcup\limits_{k=0}^{r-v-1-s} G_k.$$
	The proof is complete.
\end{proof}
\begin{corollary}
	Suppose that  $mv = r+1$ for some $v \geq 1$ and $1\leq s\leq r$.
	If $2\leq s\leq r-v-1$, let $t_{i} = m-\ceil*{\frac{(i+1)m}{r}}$ for $0\leq i\leq s$.
	Then
	$$\overline{G}_0(Q_\infty,Q_1,\cdots,Q_s) = \{(j_0,j_1,\cdots,j_s) ~|~ 1\leq j_{\sigma(i)}\leq t_{i} \text{ for } 0\leq i\leq s,~\sigma\in S_{s+1}\},$$
	where $S_{s+1}$ is the set of permutations of the set $\{0,1,\cdots,s\}$.

	If $r-v\leq s \leq r$, then $\overline{G}_0(Q_\infty,Q_1,\cdots,Q_s) = \varnothing$.
\end{corollary}
\begin{proof}
	If $2\leq s\leq r-v-1$, for each $0\leq k\leq r-v-1-s$, let $t_{i,k} = m-\ceil*{\frac{m(k+i+1)}{r}}$ for $0\leq i\leq s$.
	Let $d = r-v-1-s$. According to Proposition \ref{precise_puregap2}, we have 
	\begin{align*}
		\overline{G}_0(Q_\infty,Q_1,\cdots,Q_s) &= \bigcup\limits_{k=0}^{d} \{(j_0,j_1,\cdots,j_s)~|~1\leq j_{\sigma(i)}\leq t_{i,k} \text{ for } 0\leq i\leq s, ~\sigma\in S_{s+1}\}\\
		&= \{(j_0,j_1,\cdots,j_s)~|~1\leq j_{\sigma(i)}\leq t_{i,0} \text{ for } 0\leq i\leq s, ~\sigma\in S_{s+1}\}\\
		&= \{(j_0,j_1,\cdots,j_s) ~|~ 1\leq j_{\sigma(i)}\leq t_{i} \text{ for } 0\leq i\leq s,~\sigma\in S_{s+1}\},
	\end{align*}
	where $S_{s+1}$ is the set of permutations of the set $\{1,\cdots,s\}$.
	
	If $r-v\leq s \leq r$, then $\overline{G}_0(Q_\infty, Q_1,\cdots,Q_s) = \varnothing$ since $G_0(Q_\infty,Q_1,\cdots,Q_s) = \varnothing$.
\end{proof}

\subsection{\texorpdfstring{Pure gaps at $Q_\infty, Q_1,\cdots,Q_s$ for arbitrary $m$ and $r$}{Pure gaps at Q0, Q1, ..., Qs for arbitrary m and r}}
In the previous subsection, in the case where $mv = r+1$ for some $v\geq 1$, we have proved Proposition \ref{precise_puregap2}, which presents an explicit description of $G_0(Q_\infty,Q_1,\cdots,Q_s)$ for $1\leq s\leq r$.
Now for arbitrary $m$ and $r$, we will prove Proposition \ref{precise_puregap3}, which establishes an explicit description of $G_0(Q_\infty,Q_1,\cdots,Q_s)$ for $1\leq s\leq r$.
\begin{proposition}\label{precisefamily3}
	Suppose that $1\leq s\leq  r$. Let $1\leq k \leq r-1$ and $1\leq j_0\leq m-1$ such that $mk-rj_0\in G(Q_\infty)$.
	Let $S_s$ be the set of permutations of the set $\{1,\cdots,s\}.$
	For $1\leq i\leq s$, let 
	$$t_i = \left\{\begin{array}{cc}
		m-\ceil*{\frac{(k+i)m}{r}}+j_0,~&\text{if}~ k+i \neq r,\\
		m-\ceil*{\frac{(k+i)m}{r}}+j_0-1,~&\text{if}~ k+i = r.\\
	\end{array}\right.$$
	Then for all $\sigma\in S_s$,
	$$\{(mk-rj_0,j_1,\cdots,j_s)~|~ 1\leq j_{i}\leq t_i \text{ for } 1\leq i\leq s\}\subseteq G_0(Q_\infty,Q_{\sigma(1)},\cdots,Q_{\sigma(s)}).$$
\end{proposition}
\begin{proof}
	For each $1\leq i\leq s$, it is obvious that $1-m\leq 1-j_0\leq j_{i}-j_0\leq m-\ceil*{\frac{(k+i)m}{r}}\leq m-1$.
	Then by Lemma \ref{equivalent}, we have 
	\begin{align}\label{equivalent5}
		\floor*{\frac{(j_i-j_0)r}{m}}\leq r-(k+i)-1 \text{ if and only if } j_i-j_0 \leq t_{i} - j_0.
	\end{align}
	We set $G =\{(mk-rj_0,j_1,\cdots,j_s)~|~ 1\leq j_i\leq t_i \text{ for } 1\leq i\leq s\}$.
	For any $(mk-rj_0,j_1,\cdots,j_s)\in G$ and $\sigma\in S_s$, we already konw that $mk-rj_0\in G(Q_\infty)$.
	We will show that $(mk-rj_0,j_1,\cdots,j_l)\in G_0(Q_\infty,Q_{\sigma(1)},\cdots,Q_{\sigma(l)})$ for $1\leq l\leq s$ by induction.
	Let $\lambda^\prime\in \Z$ be the inverse of $\lambda$ modulo $m$.
	For each $1\leq j_1\leq t_1$, by Lemma \ref{sum_mod}, we have
	\begin{align*}
		&~\left\lfloor\frac{mk-rj_0+j_1r\la\la^\prime}{m}\right\rfloor+\left\lfloor\frac{j_1-j_1\la\la^\prime}{m}\right\rfloor +\sum_{i=2}^{r}\left\lfloor\frac{-j_1\la\la^\prime}{m}\right\rfloor\\
		=&~ k + \floor*{\frac{(j_1-j_0)r}{m}} +\left\lfloor\frac{j_1-j_1}{m}\right\rfloor + \sum_{i=2}^{r}\left\lfloor\frac{-j_1}{m}\right\rfloor = k+\floor*{\frac{(j_1-j_0)r}{m}} - r+1.
	\end{align*}
	It follows from (\ref{equivalent5}) that $\floor*{\frac{(j_1-j_0)r}{m}}\leq r-k-1-1$, and then $k+\floor*{\frac{(j_1-j_0)r}{m}} - r+1\leq -1$.
	Therefore, by Theorem \ref{find_another_pure_gap}, we have $(mk-rj_0,j_1)\in G_0(Q_\infty,Q_{\sigma(1)})$ .
	
	Assume that $(mk-rj_0,j_1,\cdots,j_{l-1})\in G_0(Q_\infty,Q_{\sigma(1)},\cdots,Q_{\sigma(l-1)})$ for some $2\leq l\leq s$.
	For each $1\leq j_l \leq t_l$, by Lemma \ref{sum_mod}, we have
	\begin{align*}
		&\left\lfloor\frac{mk-rj_0+j_lr\la\la^\prime}{m}\right\rfloor + \sum_{i=1}^{l}\left\lfloor\frac{j_i-j_l\la\la^\prime}{m}\right\rfloor + \sum_{i=l+1}^{r}\left\lfloor\frac{-j_l\la\la^\prime}{m}\right\rfloor\\
	  = ~&k + \floor*{\frac{(j_l-j_0)r}{m}} + \sum_{i=1}^{l}\left\lfloor\frac{j_i-j_l}{m}\right\rfloor +\sum_{i=l+1}^{r}\left\lfloor\frac{-j_l}{m}\right\rfloor\leq k + \floor*{\frac{(j_l-j_0)r}{m}} -r+l.
	\end{align*}
	It follows from (\ref{equivalent5}) that $\floor*{\frac{(j_l-j_0)r}{m}}\leq r-k-l-1$, and then $ k + \floor*{\frac{(j_l-j_0)r}{m}} -r+l\leq -1$.
	Therefore, by Theorem \ref{find_another_pure_gap}, we have 
	$$(mk-rj_0,j_1,\cdots,j_l)\in G_0(Q_\infty,Q_{\sigma(1)},\cdots,Q_{\sigma(l)}).$$
	By induction, we obtain that $(mk-rj_0,j_1,\cdots,j_s)\in G_0(Q_\infty,Q_{\sigma(1)},\cdots,Q_{\sigma(s)})$.
	Thus we have $G\subseteq G_0(Q_\infty,Q_{\sigma(1)},\cdots,Q_{\sigma(s)})$.
\end{proof}

Based on the above proposition, now we finish the proof of Proposition \ref{precise_puregap3} as follows.
\begin{proof}[Proof of Proposition \ref{precise_puregap3}]
		Let $1\leq j_0\leq c$. We define
	$$A_{j_0} = \left\{(mk_0-rj_0,a_1,\cdots,a_s)\in G_0(Q_\infty,Q_1,\cdots,Q_s)~\Big|~\ceil*{\frac{rj_0}{m}}\leq k_0\leq r-1\right\}.$$
	If $1\leq s\leq r-\floor*{\frac{r}{m}}-1$, for each $\ceil*{\frac{rj_0}{m}}\leq k\leq r-1$, we set $G_k = \{(mk_0-rj_0,mk_1+j_1,\cdots,mk_s+j_s)~|~k_0\geq\ceil*{\frac{rj_0}{m}},~
	k_i\geq 0 \text{ for } 0\leq i\leq s,~\sum_{i=0}^s k_i = k,~ 1\leq j_{\sigma(i)}\leq t_{j_0,i,k} \text{ for } 0\leq i\leq s,~ \sigma\in S_{s}\}$
	and  $\hat{G}_k = \{(mk-rj_0,j_1,\cdots,j_s)~|~1\leq j_{\sigma(i)}\leq t_{j_0,i,k} \text{ for } 1\leq i\leq s,~\sigma\in S_{s}\}$.
	It is sufficient to show that 
	\begin{align*}
		A_{j_0} = \bigcup_{k=\ceil*{\frac{rj_0}{m}}}^{d_{j_0}}G_k.
	\end{align*}
	It follows from Lemma \ref{special_gap} that $mk_0-rj_0\in G(Q_\infty)$ for $\ceil*{\frac{rj_0}{m}}\leq k_0\leq d_{j_0}$. According to Proposition \ref{precisefamily3}, we have $\hat{G}_k \subseteq G_0(Q_\infty, Q_1,\cdots,Q_s)$.
	By Theorem \ref{bottom_set}, we have $G_k\subseteq A_{j_0}$ for $\ceil*{\frac{rj_0}{m}}\leq k\leq d_{j_0}$.
	Next we will show that  $$A_{j_0}\subseteq \bigcup\limits_{k=\ceil*{\frac{rj_0}{m}}}^{d_{j_0}} G_k.$$

	For each $1\leq i\leq s$ and $\ceil*{\frac{rj_0}{m}}\leq k\leq d_{j_0}$, it is obvious that $1-m\leq 1-j_0\leq j_{\sigma(i)}-j_0\leq m-\ceil*{\frac{(k+i)m}{r}}\leq m-1$.
	Then by Lemma \ref{equivalent}, we have
	\begin{align}\label{equivalent6}
	\floor*{\frac{(j_{\sigma(i)}-j_0)r}{m}}\leq r-(k+i)-1 \text{ if and only if } j_{\sigma(i)}-j_0 \leq t_{j_0,i,k} - j_0.
	\end{align}

	For any $(mk_0-rj_0,a_1,\cdots,a_s)\in A_{j_0}$,
	by Propositions \ref{subpuregap} and \ref{novel_gap}, we can write $a_i = mk_i+j_i$ with $0\leq k_i\leq r-2-\floor*{\frac{r}{m}}$ and  $1\leq j_i\leq m-\ceil*{\frac{(k_i+1)m}{r}}$ for $1\leq i\leq s$.
	There exists some $\sigma\in S$ such that $j_{\sigma(1)}\geq j_{\sigma(2)}\geq \cdots\geq j_{\sigma(s)}$.
	Let $\lambda^\prime\in \Z$ be the inverse of $\lambda$ modulo $m$ and let $r^\prime\in\Z$ be the inverse of $-r\la$ modulo $m$.
	
	For each $1\leq l\leq s$, by Lemma \ref{sum_mod}, we have 
	\begin{align*}
		&~\floor*{\frac{mk_0-rj_0+a_{\sigma(l)}r\la\la^\prime}{m}}+\sum_{i=1}^{s}\left\lfloor\frac{a_{\sigma(i)}-a_{\sigma(l)}\la\la^\prime}{m}\right\rfloor+ \sum_{i=s+1}^{r}\left\lfloor\frac{-a_{\sigma(l)}\la\la^\prime}{m}\right\rfloor\\
	  = &~k_0+\sum_{i=1}^{s}k_{\sigma(i)} + \floor*{\frac{(j_{\sigma(l)}-j_0)r}{m}}+ \sum_{i=1}^{s}\left\lfloor\frac{j_{\sigma(i)}-j_{\sigma(l)}}{m}\right\rfloor+\sum_{i=s+1}^{r}\left\lfloor\frac{-j_{\sigma(i)}}{m}\right\rfloor\\
	  = &~k + \floor*{\frac{(j_{\sigma(l)}-j_0)r}{m}}+ \sum_{i=l+1}^{s}\left\lfloor\frac{j_{\sigma(i)}-j_{\sigma(l)}}{m}\right\rfloor-r+s.
	\end{align*}
	Let $k = \sum_{i=0}^{s}k_i$. When $l = s$, by Lemma \ref{simple_criterion}, we must have 
	$$k+\floor*{\frac{(j_{\sigma(s)}-j_0)r}{m}}-r+s\leq -1,$$
	which is equivalent to $k\leq r-\floor*{\frac{(j_{\sigma(s)}-j_0)r}{m}}-s-1$.
	Since $j_{\sigma(s)}\geq 1$, we obtain that 
	$$k\leq r-\floor*{\frac{(1-j_0)r}{m}}-s-1 = d_{j_0}.$$
	Note that $k\geq \ceil*{\frac{rj_0}{m}}$, $m \nmid rj_0$ and $m\nmid r(j_0-1)$. we have
	\begin{align*}
		s\leq r-k-\floor*{\frac{(1-j_0)r}{m}}-1&\leq r-\ceil*{\frac{rj_0}{m}}+\ceil*{\frac{r(j_0-1)}{m}}-1\\
		& = r-\left(\floor*{\frac{rj_0}{m}} - \floor*{\frac{r(j_0-1)}{m}}\right)-1.
	\end{align*}
	By Lemma \ref{floor_minus}, we obtain that $s\leq r-\floor*{\frac{r}{m}}-1$.
	Thus we have $G_0(Q_\infty,Q_1,\cdots,Q_s) = \varnothing$ for $r-\floor*{\frac{r}{m}}\leq s\leq r$.
	
	Again by Lemma \ref{simple_criterion}, we have
	$$k+\floor*{\frac{(j_{\sigma(l)}-j_0)r}{m}}+\sum_{i=l+1}^{s}\left\lfloor\frac{j_{\sigma(i)}-j_{\sigma(l)}}{m}\right\rfloor-r+s\leq -1,$$
	which is equivalent to 
	$$\floor*{\frac{(j_{\sigma(l)}-j_0)r}{m}}\leq r-k-s-1-\sum_{i=l+1}^{s}\left\lfloor\frac{j_{\sigma(i)}-j_{\sigma(l)}}{m}\right\rfloor.$$
	Since $\sum_{i=l+1}^{s}\left\lfloor\frac{j_{\sigma(i)}-j_{\sigma(l)}}{m}\right\rfloor\geq -(s-l)$, we have 
	$$\floor*{\frac{(j_{\sigma(l)}-j_0)r}{m}}\leq r-k-s-1+(s-l) = r-k-l-1.$$
	It follows from (\ref{equivalent6}) that $(mk_0-rj_0,a_1,\cdots,a_s)\in G_k$, and then 
	$$A_{j_0}\subseteq \bigcup\limits_{k=\ceil*{\frac{rj_0}{m}}}^{d_{j_0}} G_k.$$
	The proof is complete.
\end{proof}

\begin{corollary}
	Suppose that $1\leq s\leq r$.
	If $1\leq s\leq r-\floor*{\frac{r}{m}}-1$,
	for each  $1\leq j_0\leq r-1-\floor*{\frac{m}{r}}$, let 
	$$t_{j_0,i} = \left\{\begin{array}{cc}
		m-\ceil*{\frac{\left(\ceil*{\frac{rj_0}{m}}+i\right)m}{r}}+j_0,~&\text{if}~ \ceil*{\frac{rj_0}{m}}+i \neq r,\\
		m-\ceil*{\frac{\left(\ceil*{\frac{rj_0}{m}}+i\right)m}{r}}+j_0-1,~&\text{if}~ \ceil*{\frac{rj_0}{m}}+i = r,\\
	\end{array}\right.$$ for all $1\leq i\leq s$, and
	$G_{j_0} = \{(m\ceil*{\frac{rj_0}{m}}-rj_0,j_1,\cdots,j_s)~|~ 1\leq j_{\sigma(i)}\leq t_{j_0,i} \text{ for } 1\leq i\leq s,~\sigma \in S_s\}$,
	where $S_s$ is the set of permutations of the set $\{1,\cdots,s\}$.
	Then
	\begin{align*}
		\overline{G}_0(Q_\infty,Q_1,\cdots,Q_s) = \bigcup_{j_0 = 1}^{m-1-\floor*{\frac{m}{r}}}G_{j_0}.
	\end{align*}

	If $r-\floor*{\frac{r}{m}}\leq s\leq r$, then $\overline{G}_0(Q_\infty,Q_1,\cdots,Q_s) = \varnothing$.
\end{corollary}
\begin{proof}
	If $1\leq s\leq r-\floor*{\frac{r}{m}}-1$,
	for each  $1\leq j_0\leq m-1-\ceil*{\frac{m}{r}}$ and $\ceil*{\frac{rj_0}{m}}\leq k\leq r-1-\sum_{i=1}^{s}\floor*{\frac{1-j_0}{m}}-s$, let 
	$$t_{j_0,i,k} = \left\{\begin{array}{cc}
		m-\ceil*{\frac{(k+i)m}{r}}+j_0,~&\text{if}~ k+i \neq r,\\
		m-\ceil*{\frac{(k+i)m}{r}}+j_0-1,~&\text{if}~ k+i = r,\\
	\end{array}\right.$$ for all $1\leq i\leq s$, and $G_{j_0,k} = \{(m\ceil*{\frac{rj_0}{m}}-rj_0,j_1,\cdots,j_s)~|~1\leq j_{\sigma(i)}\leq t_{j_0,i,k} \text{ for } 1\leq i\leq s,~\sigma \in S_s\}$.
	According to Proposition \ref{precise_puregap3}, we have
	\begin{align*}
		\overline{G}_0(Q_\infty,Q_1,\cdots,Q_s) &= \bigcup_{j_0 = 1}^{m-1-\ceil*{\frac{m}{r}}}\bigcup_{k=\ceil*{\frac{rj_0}{m}}}^{r-\floor*{\frac{(1-j_0)r}{m}}-s-1}G_{j_0,k} \\
		&= \bigcup_{j_0 = 1}^{m-1-\floor*{\frac{m}{r}}} G_{j_0,\ceil*{\frac{rj_0}{m}}}=\bigcup_{j_0 = 1}^{m-1-\floor*{\frac{m}{r}}} G_{j_0}.
	\end{align*}
	If $r-\floor*{\frac{r}{m}}\leq s\leq r$, then $\overline{G}_0(Q_\infty,Q_1,\cdots,Q_s) = \varnothing$ since $G_0(Q_\infty,Q_1,\cdots,Q_s) = \varnothing$.
\end{proof}

Pure gaps obtained in Proposition \ref{precisefamily3} are not consecutive in the position of $Q_\infty$.
However, in the case where $m  = ur+1$ for some $u\geq 1$, we are able to obtain a family of pure gaps which are also consecutive in the position of $Q_\infty$.
The following proposition generalizes the result in \cite[Proposition 4.5]{bartoliAlgebraicGeometricCodes2018}.
\begin{proposition}\label{precisefamily4}
	Suppose that $m = ur+1$ for some $u\geq 1$. Let $1\leq s\leq r-2$ and $0\leq c\leq u(r-s-1)-1$.
	Let $t_i = u(r-i)-c-1$ for $1\leq i\leq s$.
	Then
	\begin{align*}
	\{(rc+k,a_1,\cdots,a_s)~|~r-s-1\leq k\leq r-1,~ 1\leq a_i\leq t_i &\text{ for } 1\leq i\leq s\}\\
	&\subseteq G_0(Q_\infty,Q_1,\cdots, Q_s).
	\end{align*}
\end{proposition}
\begin{proof}
	Suppose that $m = ur+1$ for some $u\geq 1$.
	Note that $mk-rj = (uk-j)r+k$.
	Then by Proposition \ref{special_gap}, we have 
	$G(Q_\infty) = \{(uk-j)r+k~|~1\leq k\leq r-1,~1\leq j\leq uk\}$.
	Let $1\leq s\leq r-2$ and $0\leq c\leq u(r-s-1)-1$.
	For each $r-s-1\leq k\leq r-1$, let $j = uk-c$.
	Since $m-\ceil*{\frac{(k+i)m}{r}}+j= ur+1 - u(k+i)-\ceil*{\frac{k+i}{r}}+j = u(r-i)-\ceil*{\frac{k+i}{r}}-c+1$ for $1\leq i\leq s$, we have
	$$t_i \leq \left\{\begin{array}{cc}
		m-\ceil*{\frac{(k+i)m}{r}}+j,~&\text{if}~ k+i \neq r,\\
		m-\ceil*{\frac{(k+i)m}{r}}+j-1,~&\text{if}~ k+i = r,\\
	\end{array}\right.$$
	Then by Proposition \ref{precisefamily3}, we have 
	$$\{(rc+k,a_1,\cdots,a_s)~|~ 1\leq a_i\leq t_i \text{ for } 1\leq i\leq s\}\subseteq G_0(Q_\infty,Q_1,\cdots, Q_s).$$
	The proof is complete.
\end{proof}

\section{Multi-point codes from Kummer extensions}\label{sec5}

In this section, we provide several examples of multi-point differential AG codes to llustrate our results.
All the codes in our examples have good parameters according to MinT's Tables \cite{mintOnlineDatabaseOptimal}.
Firstly, we consider a function field studied in \cite[Example 2]{garzonBasesRiemannRoch2022}, where the multiplicities are not equal.
\begin{example}\label{new_record_example}
	The function field $F = \mathbb{F}_{25}(x,y)$ defined by $y^8  = (x+1)^3(x^2+x+2)^7$ has genus $g = 7$ and the number of rational places $N = 76$.
	Let $\alpha_1 = -1$, $\alpha_2 = \beta$, where $\beta^2+\beta+2 = 0$. Using Algorithm \ref{alg3}, we have
	$$G_0(Q_1,Q_2) =  \left\{\begin{array}{l}
		(1,1),(1,2),(1,3),(1,4),(2,1),(2,2),(2,3),(2,4),(2,9),\\
		(4,1),(4,2),(4,3),(5,1),(5,2),(5,3),(7,1),(7,2),(10,1)
	\end{array}\right\}.$$
	Take $G_1 = 3Q_1+17Q_2$ and $G_2 =  19Q_1+Q_2$. Consider the divisor $D$ as the sum of $n$ rational places with exceptions of $Q_1$ and $Q_2$, where $69\leq n\leq 74$.
	We obtain the codes $C_\Omega(D,G_1)$ and  $C_\Omega(D,G_2)$ with parameters $[n,n-14, \geq 10]$, which both exceed the minimum distance according to MinT's Tables \cite{mintOnlineDatabaseOptimal}.
	
	In \cite[Example 2]{garzonBasesRiemannRoch2022}, the authors also present codes with parameters $[n,n-14, \geq 10]$ for $69\leq n\leq 73$. 
	Here, the length of our codes can attain 74. We obtain a new $[74, 60, \geq 10]$-code with the best known parameters.
\end{example}

Subsequently, we construct three families of multi-point differential AG codes based on the findings obtained in the previous section.
The following three propositions can be easily derived from Theorem \ref{construction} based on Propositons \ref{precisefamily1}, \ref{precisefamily2} and \ref{precisefamily4}, respectively.
We only prove Proposition \ref{construction1}. The proofs of Propositons \ref{construction2} and \ref{construction3} are similar, and they will be omitted.

\begin{proposition}\label{construction1}
	Suppose that the function field $F = K(x,y)$ defined by (\ref{equation2}) has genus $g$ and the number of rational places $N$.
	Let $2\leq s\leq r-\floor*{\frac{r}{m}}-1$ and $0\leq k\leq r-\floor*{\frac{r}{m}}-1-s$.
	Let $Q_1,\cdots,Q_s$ be $s$ totally ramified places distinct to $Q_\infty$.
	Define the divisors
	$$ G = \left((2k+1)m - \ceil*{\frac{(k+1)m}{r}}\right)Q_1+\sum_{i=2}^s \left(m-\ceil*{\frac{(k+i)m}{r}}\right)Q_i \text{,~and}$$
	$$ D = \sum_{\substack{P\in \P_1(F)\\P\notin\{Q_1,\cdots,Q_s\}}}P.$$
	If $2g-2<\deg(G)<n\leq N-s$, then the differential AG code $C_{\Omega}(D,G)$ has parameters $[N-s,k_\Omega,d_\Omega]$, where
	$$k_\Omega =n+g-1-\deg G{,~~}d_\Omega \geq \deg G-(2g-2) + sm - \sum_{i=1}^s \ceil*{\frac{(k+i)m}{r}}.$$
\end{proposition}
\begin{proof}
	By Proposition \ref{precisefamily1}, we know that
	$$\left\{(mk+j_1,\cdots,j_s)~|~ 1\leq j_{i}\leq m-\ceil*{\frac{(k+i)m}{r}} \text{ for } 1\leq i\leq s\right\}\subseteq G_0(Q_{1},\cdots, Q_{s}).$$
	Let $(a_1,a_2,\cdots,a_s) = (mk+1,1,\cdots,1)$ and $(b_1,b_2,\cdots,b_s) = (mk+m-\ceil*{\frac{(k+1)m}{r}},m-\ceil*{\frac{(k+2)m}{r}},\cdots,m-\ceil*{\frac{(k+s)}{r}})$.
	Define the divisors 
	$$G =\sum_{i=1}^{s}(a_i+b_i-1)Q_i =  \left((2k+1)m - \ceil*{\frac{(k+1)m}{r}}\right)Q_1+\sum_{i=2}^s \left(m-\ceil*{\frac{(k+i)m}{r}}\right)Q_i,$$
	$$\text{ and }D = \sum_{\substack{P\in \P_1(F)\\P\notin\{Q_1,\cdots,Q_s\}}}P.$$
	According to Theorem \ref{construction}, if $2g-2<\deg(G)<n\leq N-s$, we have the differential AG code $C_{\Omega}(D,G)$ has parameters $[n,k_\Omega,d_\Omega]$, where
	$k_\Omega = n+g-1-\deg G$, and 
	$$d_\Omega\geq \deg G-(2g-2) + s + \sum_{i=1}^s(b_i-a_i) = \deg G-(2g-2) + sm - \sum_{i=1}^s \ceil*{\frac{(k+i)m}{r}}.$$
\end{proof}
\begin{proposition}\label{construction2}
	Suppose that the function field $F = K(x,y)$ defined by (\ref{equation2}) has genus $g$ and the number of rational places $N$.
	Suppose that $mv=r+1$ for some $v\geq 1$.
	Let $1\leq s \leq r-1-v$ and $0\leq k\leq r-v-s-1$.
	Let $Q_1,\cdots,Q_s$ be $s$ totally ramified places distinct to $Q_\infty$.
	Define the divisors
	$$G = \left((2k+1)m-\ceil*{\frac{(k+1)m}{r}}\right)Q_\infty + \sum_{i=1}^s \left(m-\ceil*{\frac{(k+i+1)m}{r}}\right) Q_i,\text{~and}$$
	$$D = \sum_{\substack{P\in \P_1(F)\\P\notin\{Q_\infty,Q_1,\cdots,Q_s\}}}P.$$
	If $2g-2<\deg(G)<n\leq N-s-1$, then the differential AG code $C_{\Omega}(D,G)$ has parameters $[n,k_\Omega,d_\Omega]$, where
	$$k_\Omega = n+g-1-\deg G{,~~}d_\Omega \geq \deg G-(2g-2) + (s+1)m -\sum_{i=0}^s\ceil*{\frac{(k+i+1)m}{r}}.$$
\end{proposition}

\begin{proposition}\label{construction3}
	Suppose that the function field $F = K(x,y)$ defined by (\ref{equation2}) has genus $g$ and the number of rational places $N$.
	Suppose that $m=ur+1$ for some $u\geq 1$.
	Let $1\leq s\leq r-2$ and $0\leq c\leq u(r-s-1)-1$.
	Let $Q_1,\cdots,Q_s$ be  $s$ totally ramified places distinct to $Q_\infty$.
	Define the divisors
	$$G = \left((2c+2)r-s-3\right)Q_\infty + \sum_{i=1}^s \left(u(r-i)-c-1\right) Q_i,\text{~and}$$
	$$D = \sum_{\substack{P\in \P_1(F)\\P\notin\{Q_\infty,Q_1,\cdots,Q_s\}}}P.$$
	If $2g-2<\deg(G)<n\leq N-s-1$, then the differential AG code $C_{\Omega}(D,G)$ has parameters $[n,k_\Omega,d_\Omega]$, where
	$$k_\Omega = n+g-1-\deg G,~~d_\Omega \geq \deg G-(2g-2)+1+(ur-c)s-\frac{s(s+1)u}{2}.$$
\end{proposition}
\begin{example}\label{example1}
	Consider the function field $F = \mathbb{F}_{q^t}(x,y)$ defined by $y^m  = (x^{q^{t/2}}-x)^{q^{t/2}-1}$, where $t$ is even, $q$ is a prime power, $m\mid (q^t-1)$ and $\gcd(m,q^{t/2}-1)=1$.
	From \cite{garciaConstructionCurvesFinite2001}, we know that $F$ has genus $g = (q^{t/2}-1)(m-1)/2$ and $N = (q^t-q^{t/2})m+q^{t/2}+1$ rational places over $\mathbb{F}_{q^{2t}}$.
	Let $n = N-s$.
	Table \ref{table1_1} shows the parameters of the multi-point AG codes constructed using Proposition \ref{construction1}.
	Table \ref{table1_2} shows the parameters of the multi-point AG codes constructed using Proposition \ref{construction2}.
	We observe that some multi-point AG codes have better parameters compared to the corresponding ones in the MinT's Table \cite{mintOnlineDatabaseOptimal}.
\end{example}

\begin{table}[h]
	\centering
	\caption{Multi-point AG codes constructed using Proposition \ref{construction1} in Example \ref{example1}}
	\begin{tabularx}{\textwidth}{ccccccc}
		\toprule
		$q$ & $t$ & $m$ & $s$ & $k$ & $[n,k_\Omega,d_\Omega]$ & Improvement on $d_\Omega$ compared with \cite{mintOnlineDatabaseOptimal}\\ \midrule
		8 & 2 & 9 & 2 & 5 & $[511,445,\geq 42]$ & 3\\ 
		8 & 2 & 9 & 3 & 4 & $[510,459,\geq 30]$ & 2\\ 
		8 & 2 & 3 & 2 & 3 & $[175,161,\geq 10]$ & 1\\
		9 & 2 & 5 & 2 & 5 & $[368,331,\geq 24]$ & 2\\
		9 & 2 & 5 & 3 & 4 & $[367,338,\geq 18]$ & 1\\
		5 & 2 & 3 & 2 & 1 & $[64,59,\geq 4]$ & 0\\ \botrule
	\end{tabularx}
	\label{table1_1}
\end{table}
\begin{table}[h]
	\centering
	\caption{Multi-point AG codes constructed using Proposition \ref{construction2} in Example \ref{example1}}
	\begin{tabularx}{\textwidth}{ccccccc}
		\toprule
		$q$ & $t$ & $m$ & $s$ & $k$ & $[n,k_\Omega,d_\Omega]$& Improvement on $d_\Omega$ compared with \cite{mintOnlineDatabaseOptimal}\\  \midrule
		8 & 2 & 9 & 1 & 5 & $[511,445,\geq 42]$ & 3\\
		8 & 2 & 9 & 2 & 4 & $[510,459,\geq 30]$ & 2\\
		8 & 2 & 3 & 1 & 3 & $[175,161,\geq 10]$ & 1\\
		9 & 2 & 5 & 1 & 5 & $[368,331,\geq 24]$ & 2\\
		9 & 2 & 5 & 2 & 4 & $[367,338,\geq 18]$ & 1\\
		5 & 2 & 3 & 1 & 1 & $[64,59,\geq 4]$ & 0\\ \botrule
	\end{tabularx}
	\label{table1_2}
\end{table}

\begin{example}\label{example2}
	Consider the function field $F = \mathbb{F}_{q^2}(x,y)$ defined by $y^m = x^q+x$, where $q$ is a prime power and $m\mid q+1$.
	The function field $F$ has genus $g = (m-1)(q-1)/2$ and the number of rational places is $N = q(1+(q-1)m)+1$ over $\mathbb{F}_{q^2}$.
	Let $n = N-s$.
	Table \ref{table2_1} shows the parameters of the multi-point AG codes constructed using Proposition \ref{construction1}.
	Table \ref{table2_2} shows the parameters of the multi-point AG codes constructed using Proposition \ref{construction2}.
	We observe that some multi-point AG codes have better parameters compared to the corresponding ones in the MinT's Table \cite{mintOnlineDatabaseOptimal}.
\end{example}

\begin{table}[h]
	\centering
	\caption{Multi-point AG codes constructed using Proposition \ref{construction1} in Example \ref{example2}}
\begin{tabularx}{\textwidth}{cccccc}
	\toprule
	$q$ & $m$ & $s$ & $k$ & $[n,k_\Omega,d_\Omega]$& Improvement on $d_\Omega$ compared with \cite{mintOnlineDatabaseOptimal}\\  \midrule
	5 & 6 & 2 & 2 & $[124,106,\geq 12]$ & 1\\
	7 & 8 & 2 & 4 & $[342,295,\geq 30]$ & 3\\
	7 & 4 & 2 & 3 & $[174,156,\geq 12]$ & 1\\
	8 & 9 & 2 & 5 & $[511,445,\geq 42]$ & 3\\
	8 & 3 & 2 & 3 & $[175,161,\geq 10]$ & 1\\
	9 & 5 & 2 & 5 & $[368,331,\geq 24]$ & 2\\
	9 & 5 & 3 & 4 & $[367,338,\geq 18]$ & 1\\
	9 & 2 & 2 & 2 & $[152,145,\geq 6]$ & 0\\ \botrule
\end{tabularx}
\label{table2_1}
\end{table}

\begin{table}[h]
	\centering
	\caption{Multi-point AG codes constructed using Proposition \ref{construction2} in Example \ref{example2}}
\begin{tabularx}{\textwidth}{cccccc}
	\toprule
	$q$ & $m$ & $s$ & $k$ & $[n,k_\Omega,d_\Omega]$& Improvement on $d_\Omega$ compared with \cite{mintOnlineDatabaseOptimal}\\ \midrule
	5 & 6 & 1 & 2 & $[124,106,\geq 12]$ & 1\\ 
	7 & 8 & 1 & 4 & $[342,295,\geq 30]$ & 3\\ 
	7 & 4 & 1 & 3 & $[174,156,\geq 12]$ & 1\\ 
	8 & 9 & 1 & 5 & $[511,445,\geq 42]$ & 3\\ 
	8 & 3 & 1 & 3 & $[175,161,\geq 10]$ & 1\\
	9 & 5 & 1 & 5 & $[368,331,\geq 24]$ & 2\\ 
	9 & 5 & 2 & 4 & $[367,338,\geq 18]$ & 1\\ 
	9 & 2 & 1 & 2 & $[152,145,\geq 6]$ & 0\\ \botrule
\end{tabularx}
\label{table2_2}
\end{table}

Some AG codes presented in Tables \ref{table1_1} and \ref{table2_1} can also be found in \cite{castellanosSetPure2024}, where the infinite place was not taken into consideration.
Here, we incorporate the infinite place in the construction of AG codes and obtain Tables \ref{table1_2} and \ref{table2_2}.
It is observed that the parameters of the AG codes in Tables \ref{table1_1} and \ref{table1_2} are identical,
and the same holds true for the parameters of the AG codes in Tables \ref{table2_1} and \ref{table2_2}.
Finally, we present an example in accordance with Proposition \ref{construction3}.

\begin{example}
	Consider the function field $F = \mathbb{F}_{q^2}(x,y)$ defined by $y^{q+1} = \sum_{i=1}^t x^{q/2^i}$, where $q = 2^t$ and $t\geq 2$.
	The function field $F$ has genus $g = q(q-2)/2$ and the number of rational places is $N = 1+q^2+2gq$ over $\mathbb{F}_{q^2}$ \cite{abdnMaximalCurves1999}.
	When $t = 3$, we have the function field $F = \mathbb{F}_{64}(x,y)$ defined by $y^{9} = x^4+x^2+x$.
	Take $s=1$ and $c=3$ in Proposition \ref{construction3}.
	Then we obtain a multi-point AG code with parameters $[255,236,\geq 12]$,  which attains the minimum distance according to MinT's Tables \cite{mintOnlineDatabaseOptimal}.
\end{example}

\section{Acknowledgements}
This work is supported by Guangdong Basic and Applied Basic Research Foundation(No. 2025A1515011764), the National Natural Science Foundation of China (No. 12441107), 
Guangdong Major Project of Basic and Applied Basic Research (No. 2019B030302008) and Guangdong Provincial Key Laboratory of Information Security Technology (No. 2023B1212060026). 
\bibliography{refs}
\end{document}